\documentclass[journal,onecolumn]{IEEEtran}
\usepackage{amsmath}
\usepackage{amsthm}
\usepackage{amssymb}

\newtheorem{theorem}{Theorem}

\theoremstyle{definition}
\newtheorem{defn}{Definition}

\usepackage{enumerate}
\usepackage{svg}
\usepackage{graphicx}

\usepackage[justification=centering]{caption}

\usepackage{makecell}
\usepackage[sorting=none]{biblatex}
\usepackage{multirow}
\usepackage[flushleft]{threeparttable} 
\usepackage{hyperref}
\usepackage{xpatch}
\makeatletter
\AtBeginDocument{\xpatchcmd{\@thm}{\thm@headpunct{.}}{\thm@headpunct{}}{}{}}
\makeatother

\ifCLASSINFOpdf
\else
\fi
\begin{document}
%
\title{A Practical and Privacy-Preserving Framework for Real-World Large Language Model Services}
%
%
%

\author{Xueping~Liao, Yu~Mao, Wei~Liu, Anjia~Yang
\thanks{Yu~Mao and Xueping~Liao contributed equally to this work.}
\thanks{Xueping Liao, Yu~Maoand Wei Liu are with the School of Computer, Electronics and Information, Guangxi University, Nanning, Guangxi 530004, China (email: patrickymao@gmail.com; lxp2512001529@163.com; weiliuscholar@gmail.com).}
\thanks{Anjia Yang is with the College of Cyber Security, Jinan University, Guangzhou 510632, China (email: anjiayang@gmail.com).}
}

\maketitle

\begin{abstract}
Large language models (LLMs) have demonstrated exceptional text understanding and generation capabilities, and they are increasingly being used to improve productivity across various domains. However, owing to the high costs of training and maintaining these models, as well as the fact that some LLMs are proprietary, individuals frequently rely on online AI as a Service (AIaaS) provided by LLM companies. This business model poses significant privacy risks, because service providers may exploit users' trace patterns and behavioural data. In this study, we propose a practical and privacy-preserving framework that ensures user anonymity by preventing service providers from linking requests to the individuals who submit them. Our framework is based on partially blind signatures, which ensure the unlinkability of user requests. Furthermore, we introduce two strategies that are tailored to both subscription-based and API-based service models, ensuring the protection of both users' privacy and service providers' interests. The framework is designed to integrate seamlessly with existing LLM systems because it does not require modifications to the underlying architectures. Experimental results show that our framework incurs minimal computation and communication overhead, making it a feasible solution for real-world applications.
\end{abstract}


%
\IEEEpeerreviewmaketitle

\section{Introduction}

\begin{sloppypar}
\end{sloppypar}
The rapid advancements in training architectures and data have led to the remarkable performance of large language models (LLMs) in content analysis and generation. Given these capabilities, LLMs have been widely used across various domains, such as medicine \cite{singhal2022largelanguagemodelsencode}, robotics \cite{yu2023scalingrobotlearningsemantically, hu2023lookleapunveilingpower} and education \cite{zhang2024simulatingclassroomeducationllmempowered}, to improve productivity and efficiency. However, the high costs associated with training and maintaining these models (e.g., LLaMa-3 is pre-trained on 15 trillion multilingual tokens \cite{dubey2024llama3herdmodels}), combined with the fact that some LLMs are not publicly accessible, present significant challenges for individuals and small enterprises aiming to develop and sustain their own LLMs. Consequently, many companies specialising in LLMs have begun providing online LLM services, which is frequently referred to as AI as a Service (AIaaS) (e.g., OpenAI's ChatGPT \footnote{OpenAI ChatGPT: https://openai.com/chatgpt/} and Google's Gemini \footnote{Google Gemini: https://gemini.google.com/}).

Despite the benefits of AIaaS, using such services raises serious privacy concerns, particularly regarding the potential exposure of user information. Service providers may log information related to user requests, such as the identities of the requesters, time of the request, and response. This information could be used to trace queries back to specific users, compromising users' privacy. Hence, there is a pressing requirement for anonymity in LLM services, as is common in other online services. One prevalent approach to privacy protection in LLM services is to encrypt user request data, rendering the servers incapable of interpreting the content of the requests. For homomorphic encryption (HE), the research includes THE-X \cite{chen2022thexprivacypreservingtransformerinference}, Iron \cite{hao2022iron}, and BumbleBee \cite{cryptoeprint:2023/1678}. For secure multi-party computation (MPC), representative methods include East \cite{ding2023eastefficientaccuratesecure} and Privformer \cite{10190506}. In addition, Sigma \cite{gupta2023sigma} uses secret sharing (SS) to secure the inference process. However, these methods are typically resource-intensive, requiring significant computational demands and communication overhead, as well as modifications to the existing architectures to support encryption. Moreover, while these techniques prevent the service provider from learning the content of the queries, side-channel information, such as the timing of requests and responses, can still inadvertently leak user activity patterns.

To address these issues, a practical, secure framework for LLM services that guarantees user anonymity must satisfy several key properties. 
First, the framework should prevent service providers from linking requests to individual users, thereby preserving anonymity.
Second, it should be sufficient to be compatible with existing LLM architectures, allowing practical applicability and seamless integration with other system modules. 

In this study, we propose a framework that meets these requirements. Our framework considers the existing LLM service as a black box, acting as an additional layer that is fully transparent to the underlying system, requiring no changes to the existing infrastructure. In addition, the framework uses blind signatures to ensure complete anonymity of requests, preventing service providers from tracing requests back to their originators.

Specifically, the system uses partially blind signatures and tailored strategies to address the two prevalent business modes in LLM services. For the subscription mode, a partially blind signature-based scheme with a restricted subscription period enables users to make unlimited requests within the specified timeframe. In contrast, the API-based mode uses a different partially blind signature-based scheme that incorporates a request-level charging strategy, limiting users to a specific number of requests or resources (e.g., a predefined number of tokens sent in requests).

The contribution of this framework is as follows: 
\begin{itemize} 
  \item The framework provides a practical solution for maintaining user anonymity in LLM services while accommodating existing business models. 
  \item The proposed system can be seamlessly integrated into the current systems with minimal computation and communication overhead. 
  \item The framework demonstrates the application of partially blind signatures combined with tailored strategies in real-world scenarios, demonstrating its potential for extension to other online services. 
  \item Experimental results on the framework are provided, measuring the performance of partially blind signatures of the system. These results provide helpful information regarding the computation and communication costs associated with partially blind signatures.
\end{itemize}

\section{Related Work}
\subsection{Language Models} 
\label{relatedwork:LM}

Language models (LMs) are statistical models that predict the probability of subsequent or missing tokens within a given context. LLMs represent a class of pre-trained language models (PLMs) that use machine learning on an unprecedented scale \cite{zhao2023surveylargelanguagemodels}. LLMs differ from other PLMs primarily in their vast training data and model sizes, providing them with emergent capabilities—behaviours that arise as a result of scaling \cite{wei2022emergentabilitieslargelanguage}. These models have demonstrated exceptional performance in both content analysis and generation, frequently achieving human-like results across a variety of tasks \cite{openai2024gpt4technicalreport}. LLMs have been applied in numerous domains, including education, finance, and healthcare \cite{hadi2023survey}. In addition, a technique called prompt engineering has emerged to optimise LLM outputs for specific tasks \cite{lo2023clear, wang2023prompt}.

The development of LLMs involves three main stages: pre-training, fine-tuning, and inference.

\textbf{Pre-training:}
During pre-training, LLMs are trained on massive datasets to acquire general knowledge. These datasets typically consist of diverse sources, such as blogs, articles, and books. For example, the August 2024 CommonCrawl dataset comprised 2.30 billion web pages \footnote{CommonCrawl: https://commoncrawl.org/}. GPT-3 was pre-trained on various datasets, including filtered CommonCrawl data and Wikipedia, comprising approximately 300 billion tokens \cite{brown2020languagemodelsfewshotlearners}. Correspondingly, \mbox{LLaMA-3's} pre-training dataset contained 15 trillion multilingual tokens \cite{dubey2024llama3herdmodels}. The pre-training process was computationally intensive and time-consuming. For instance, training the LLaMA-3.1 70B model required an estimated 7 million hours on H100-80GB hardware \footnote{Hugging Face Meta-Llama-3.1-70B: https://huggingface.co/meta-llama/Meta-Llama-3.1-70B}.

\textbf{Fine-tuning:}
Following pre-training, LLMs are fine-tuned for specific tasks. The datasets used for fine-tuning are typically smaller and contain domain-specific knowledge. Reinforcement learning from human feedback (RLHF) is also commonly used during this stage to optimise the model's performance, as observed in GPT-4's training process \cite{achiam2023gpt}.

\textbf{Inference:}
Once pre-trained and fine-tuned, LLMs can be deployed for inference. However, deploying LLMs presents significant storage and computation challenges owing to their complex architectures and large model sizes, frequently consisting of billions of parameters. Consequently, running LLMs locally is typically impractical for individual users. To address these issues, several companies provide LLM services in the form of AIaaS, such as OpenAI's ChatGPT and Google's Gemini.

\subsection{Privacy Protection in the LLM Inference Stage}

In the current business model, LLM companies typically deploy trained models and expose the inference stage to users through online services. These services allow users to submit input data, and the system returns the inference results generated by the LLMs based on that data. However, this technique may expose users to potential privacy risks. To use LLMs for text analysis or generation, users frequently must provide information in their requests, which may include sensitive data.

To address these issues, various studies have focused on enhancing user privacy during LLM inference. HE is one method that allows computation on encrypted data. By integrating HE into LLMs, users can submit encrypted inputs for inference, and only the final results must be decrypted. Iron \cite{hao2022iron} and BumbleBee \cite{cryptoeprint:2023/1678} propose HE-based protocols to efficiently perform matrix multiplication, a common operation in transformers, enabling computation on encrypted data. Alternatively, THE-X \cite{chen2022thexprivacypreservingtransformerinference} develops encryption-friendly approximation techniques for components in transformers to reduce the computational overhead involved in processing encrypted data. Another cryptographic technique is MPC, where multiple parties collaboratively compute partial results based on the input. In this method, no single party has access to the full input or result. The final output is reconstructed by combining the partial results from each party. Privformer \cite{akimoto2023privformer} introduces a protocol that ensures privacy in a three-party system involving model owners, MPC servers, and users. This protocol guarantees that no single MPC server can recover the user's input data independently. Differential privacy (DP) is another prominent technique for protecting user privacy, ensuring that individual information cannot be inferred from the data. The perturbation mechanism proposed in a study \cite{majmudar2022differentiallyprivatedecodinglarge} applied DP to perturb LLM outputs, while DP-Forward \cite{du2023dp} introduced DP-forward, a technique that perturbed embedding matrices during the forward pass to ensure the security of both training and inference stages.

Although these methods provide varying levels of security, they frequently introduce significant computation and communication overhead, making inference impractical in real-world scenarios. For instance, Homomorphic-encryption-based schemes, such as THE-X \cite{chen2022thexprivacypreservingtransformerinference}, require extensive encryption, decryption, and interaction during each inference step. Because homomorphic encryption (HE) is inherently computation-intensive and demands costly approximations for non-polynomial functions, these approaches inevitably incur substantial latency overhead relative to plaintext inference on the same CPU hardware. According to \cite{hao2022iron}, when both methods adopt BERT-Base and run in WAN settings, Iron (CPU-based private inference) achieves a 65–-105$\times$ runtime speedup over MP-SPDZ, yet the overall latency remains considerable in practical deployments. BumbleBee \cite{cryptoeprint:2023/1678}, in contrast, reports a markedly higher efficiency: under comparable settings on BERT-Base, BumbleBee completes one inference in 2.55 minutes, whereas Iron requires approximately 34 minutes. As one of the currently most efficient private Transformer inference frameworks, BumbleBee still incurs notable overhead compared with plaintext inference—for example, in WAN environments, generating a single token from 128 input tokens on GPT-2 requires 3.06 minutes and transmits a total of 6.61 GB of data, far exceeding plaintext inference costs. Similarly, under a secure three-party protocol, generating a 64-word sentence can take up to 20 minutes \cite{akimoto2023privformer}.Although differential-privacy-based methods \cite{majmudar2022differentiallyprivatedecodinglarge,du2023dp} offer better computational efficiency, they inherently require a trade-off between privacy guarantees and model utility as model size increases.

\subsection{Blind Signature}
The blind signature scheme, first introduced in the study \cite{chaum1983blind}, is an interactive cryptographic protocol that enables a signer to authenticate a message without learning its content. This ensures the requester's anonymity while maintaining the integrity of the signed message. A variant of this scheme, known as the partially blind signature, was later proposed, allowing part of the message to be pre-agreed upon by both the requester and signer. This variation is particularly useful in practical scenarios where shared information, such as an expiration date or issuance timestamp, is necessary \cite{abe1996date}.

Blind signatures have been widely applied in systems designed to ensure user anonymity. Previous efforts \cite{chaum1990untraceable} used blind signatures to trace double-spending for untraceable digital cash while ensuring that transactions could not be linked back to the spender's identity. In addition, blind signatures have been integrated into numerous electronic voting systems to protect voters' identities and maintain the confidentiality of the voting process \cite{khater2018blind}. Similarly, the study \cite{ladd2012blind} proposed incorporating blind signatures into the Bitcoin protocol to unlink transaction recipients from their public keys, thereby enhancing privacy.

\section{Preliminaries}

\subsection{Public Key Infrastructure}
\label{preliminaries:pki}
The public key infrastructure (PKI) is a system designed to establish and maintain trustworthy mappings between the identities of principals (e.g., servers and individuals) and their associated public keys. These binding relationships are documented through digital certificates, which contain the principal's identifying information, public keys, and the key generation algorithms used \cite{rfc5280}. Digital certificates can be obtained from trusted entities within the PKI's trust model \cite{806987, weise2001public}. Each digital certificate is accompanied by a digital signature, created using the private key of a trusted entity, to authenticate the certificate's origin and ensure its integrity.

PKI is widely used in transport layer security (TLS) for server verification \cite{rfc8446}. During the TLS handshake, servers present their digital certificates along with the corresponding digital signatures issued by trusted entities. Clients then verify that the server's information matches the details provided in the digital certificates and validate the digital signatures. This process ensures that the connected servers are authenticated and trustworthy.

\subsection{Authentication Services}
Online services have become a prevalent business model, where users gain access to services provided by service providers once they are granted the necessary permissions. However, the inherent unreliability of the internet necessitates robust mechanisms for authenticating principals to ensure that commercial services are delivered by legitimate providers to genuine users. Authentication services can be broadly classified into two categories: anonymous and logged-in modes.

\textbf{Authentication in Anonymous Mode:}
In this model, services are provided at the request level. Users are issued certificates, such as tickets or tokens, which are registered on the servers. Services are then provided in response to requests that present these valid certificates. In this mode, users are not required to reveal their identities, provided that the certificates they present are valid. However, users must ensure that the servers handling their requests are legitimate to avoid the leakage of certificates or being spammed by fraudulent results. A common method is for users to establish a secure connection, such as hypertext transfer protocol secure (HTTPS) \cite{rfc2818}, where servers can prove their identity using secure protocols such as TLS under the PKI framework (see Section \ref{preliminaries:pki}).

\textbf{Authentication in Logged-in Mode:}
In contrast to the anonymous mode, services in this mode are provided at the user level. Users are required to prove their identities to the servers before gaining access to certain services. In contrast to the PKI system, where only public keys are stored, certain private user information, such as password hashes or phone numbers, is stored on the servers. During the login process, users prove their identity by demonstrating possession of this private information, which is only known to them. Once logged in, user information is recorded in the connection or session. Servers then retrieve the user's identity from these sessions to determine whether they are authorised to access specific services. Similar to the anonymous mode, users must also verify the identity of the servers to ensure security.

\subsection{LLM Service Models} 
\label{preliminaries:llm_services_models}
As discussed in Section \ref{relatedwork:LM}, deploying LLMs locally requires substantial computation and storage resources. In addition, some LLMs' parameters may not be publicly accessible. Consequently, most users rely on online LLM services. These services typically provide a range of pricing plans to satisfy different user requirements and can be broadly categorised into subscription-based and API-based modes.

\textbf{Subscription-Based Mode:}
In subscription-based mode, users pay a recurring fee, typically on a monthly or yearly basis, to access LLM resources within the specified time period.

\textbf{API-Based Mode:}
Conversely, API-based mode charges users based on the number of input tokens processed. This model allows for more precise cost management, allowing users to budget based on their actual usage.

\subsection{Partially Blind Signature}
Prior to formally introducing partially blind signatures, it is necessary to clarify some notations. $\lambda$ represents the security parameter, $S$ denotes the signer role, and $U$ signifies the user role. The notation $ y \leftarrow S(x)$ indicates that when the signer receives an input $x$, it produces an output $y$. A function $negl(x):\mathbb{N} \rightarrow \mathbb{R}_+$ is negligible in its input $x$ if $negl(x) \in x^{-\omega(1)}$.

Partially blind signatures differ from standard blind signatures in that they enable the signer to have partial knowledge or control over specific portions of the message while maintaining the confidentiality of the remaining content. In the subsequent sections, we will discuss the foundational components of partially blind signatures, which will be used in the protocols that follow.
\begin{defn} (Partially Blind Signature Scheme). A partially blind signature scheme $PBSS$ is a tuple of efficient algorithm $PBSS=(Setup,$ $Blind,Sign,Unblind,Verify)$:
\end{defn}

\begin{itemize}
    \item $(pk,sk)\leftarrow PBSS.Setup(\lambda)$: The setup algorithm takes a security parameter $\lambda$ as input and outputs a pair of public and private keys $(pk,sk)$.
    
    \item $M'\leftarrow PBSS.Blind(M,info)$: The blind algorithm is executed by the requester with the public metadata $info$ and plaintext message $M$ as input and outputs the blinded message $M'$.
    
    \item $\sigma' \leftarrow PBSS.Sign(M',info,sk)$: The sign algorithm is executed by the signer with the blinded message $M'$, public metadata $info$, and private key $sk$ as input and outputs the blinded signature $\sigma'$.
    
    \item $\sigma \leftarrow PBSS.Unblind(\sigma')$: The unblind algorithm is executed by the requester with the blinded signature $\sigma'$ as input and outputs the real signature $\sigma$ of the message $M$.
    
    \item $b\leftarrow PBSS.Verify(\sigma,M,info,pk)$: The verification algorithm considers the signature $\sigma$, plaintext message $M$, public metadata $info$, and public key $pk$ as input and outputs a bit b$\in \{0,1\}$.
\end{itemize}

Note that the public metadata $info$ mentioned above is mutually agreed upon by both the requester and signer \cite{abe1996date},
which is typically negotiated independently of the specific scenario, implying that the process by which $info$ is determined does not compromise the security of the scheme. Although the detailed $info$ calculation process is provided by the study \cite{Asghar2015ASO}, it is simplified in this paper while retaining the core principles. In essence, the $Ag$ function maps the negotiated public metadata string to $info$, which is shared between both parties.

\begin{defn} (Completeness). The $PBSS$ satisfies the completeness if, for all messages $M$ and public metadata $info$, the following holds:
\end{defn}
{\scriptsize
$$
    \Pr\left[\begin{array}{lll}
        &&(pk,sk)\leftarrow\text{PBSS.Setup}(\lambda)\\ 
        \text{PBSS.Verify}(\sigma,M,info,pk)=1&:&(M')\leftarrow\text{PBSS.Blind}(M,info)\\
        &&\sigma' \leftarrow\text{PBSS.Sign}(M',info,sk)\\
        &&\sigma \leftarrow\text{PBSS.Unblind}(\sigma)\end{array}\right]\geq 1-negl(\lambda)
$$
}

In addition to completeness, a fundamental property of most signature schemes, partially blind signatures exhibit two additional essential properties. The first is known as partial blindness (also referred to as unlinkability in some literature). This property ensures that, even though the interaction during the signing process, final signatures, and corresponding input messages are inspected, it remains impossible to trace back to the requests that generated the signatures. To formally define partial blindness, we first introduce GAME A.

\begin{defn}
(GAME A). We assume $U_0$ and $U_1$ are the honest users who follow the signature protocol, and $S_{adv}$ is the malicious signer:
\end{defn}

\begin{itemize}
    \item Step1. $(pk,sk)\leftarrow PBSS.Setup(\lambda)$.
    \item Step2. $(m_0,m_1,info_{0},info_{1},Ag)\leftarrow S_{adv}(sk)$.  
    \item Step3. Set up $U_0$ and $U_1$ as follows:
    \begin{itemize}
        \item Randomly select $b\in_{R} \{0,1\}$ and let $m_b$ and $m_{1-b}$ be the input of $U_0$ and $U_1$, respectively.
        \item Consider $(info_0,pk,Ag)$ and $(info_1,pk,Ag)$ as input for $U_0$ and $U_1$, respectively.
    \end{itemize}
    \item Step4. $U_0$ and $U_1$ interact with $S_{adv}$ to complete the signature protocol.
    \item Step5. If $U_0$ and $U_1$ output $\sigma(m_b)$ and $\sigma(m_{1-b})$, respectively, and $info_1=info_0$ holds, then arrange $\{\sigma(m_b),\sigma(m_{1-b})\}$ in the corresponding order of $(m_0,m_1)$ and send them to $S_{adv}$. Given $\perp$ to $S_{adv}$ otherwise.
    \item Step6. $S_{adv}$ output $b'\in \{0,1\}$ 
\end{itemize}
We assume that $S_{adv}$ wins the game if $b'=b$.

\begin{defn} (Unlinkability). A signature scheme is partially blind if, for all probabilistic polynomial-time algorithm $S_{adv}$ play the GAME A, sufficiently large $n$ and some constant $c$, the following holds:
\end{defn}
$$Pr\left[  b'=b \right]\leq \frac{1}{2}+negl(\lambda)$$

The second critical property is unforgeability, which ensures that only the legitimate signer can generate a valid signature. The definition of unforgeability for partially blind signatures is analogous to that outlined in the study \cite{Securityofblinddigitalsignatures}, with special emphasis on the role of public information. In the context of blind signatures, an adversary is considered successful if they can generate more than $l$ signatures (i.e., at least $l+1$ signatures) using at most $l$ signing requests. For partially blind signatures, the game follows a similar structure: if an adversary can produce $l+1$ signatures for any public information $info$ with no more than $l$ signing requests involving $info$, they are considered to have won the game. Before formally defining unforgeability for partially blind signatures, we introduce GAME B.

\begin{defn}
(GAME B). Let $U_{adv}$ be the malicious user and $S$ be the signer that follow the signature protocol.
\end{defn}

\begin{itemize}
    \item Step1. $( pk, sk) \leftarrow PBSS.Setup ( \lambda) .$
    \item Step2. $Ag\leftarrow U_{adv}( pk) .$  
    \item Step3. Let $sk$, $Ag$, and randomly considered $info$ as input of $S$.
    \item Step4. $U_{adv}$ engages in the signature protocol with $S$ in a concurrent manner. Let $l$ denote the number of times a signature is performed given the same $info$.
    \item Step5. $U_{adv}$ outputs the common information $info$ and $l$ signatures $(m_1,\sigma_1),...,(m_{l+1},\sigma_{l+1})$
\end{itemize}

\begin{defn}
(Unforgeability). A partially blind signature scheme is unforgeability if, for any probabilistic polynomial-time algorithm $U_{adv}$ play the GAME B, sufficiently large $\lambda$ and some constant $c$, the following holds for the output of $U_{adv}$:
$$Pr\left[ PBSS.Verify(pk,info,m_j,\sigma_j)=accept\right]\leq negl(\lambda),\ j=1,\cdots ,l+1$$
\end{defn}

\subsection{Threat Model}
We consider three scenarios involving three types of adversaries. First, in the case of a man-in-the-middle (MitM) attack between the user and server, the adversary can eavesdrop on and potentially modify the communication between the two parties. The objective of such an adversary is to gather maximum information regarding the content of requests and responses or to impersonate either the legitimate user or server. Second, if the adversary corrupts the user, their goal is to obtain additional services without payment. Finally, if the adversary corrupts the server, their aim is to gain access to sensitive information, such as the association between users' identities and their requests, or the relationships among the requests themselves.

In addition to the aforementioned scenarios, we also consider a common threat identified in computer systems—side-channel attacks. These attacks involve the adversary analysing subtle indirect information leaked by the system, such as timing variations or correlations between different operations. Although our system links different queries together, creating threads of related questions, the anonymity of the users is preserved since no login is required to submit queries. Users interact with the system using a blinded ‘token’ to make queries, which indicates that even if an attacker can correlate queries, they cannot definitively associate any query with a specific user. This mechanism ensures that the adversary cannot easily trace the source of the queries, even if they analyse patterns across multiple requests.

For potential IP-based tracking, users can access the service using The Onion Router(Tor), which provides network-layer untraceability. Tor uses onion routing with multi-layer encryption and distributed relays, ensuring that no single relay can learn both the source and destination of the communication. By combining Tor's network-layer anonymity with our application-layer unlinkability, complementary protection is achieved, improving overall anonymity.

\subsection{Design Goals}
Our objective is to design a privacy-preserving framework for LLMs that safeguards user privacy while ensuring the delivery of accurate responses. This framework is designed to provide secure communication, anonymity, and protection of users' and service providers' interests.

\textbf{Secure Communication:}
This property ensures that the messages exchanged between the user and server are protected by confidentiality, integrity, and authentication. Confidentiality guarantees that the content of the messages remains private and inaccessible to unauthorised parties. Integrity ensures that the messages are not altered or tampered with during transmission. Authentication verifies that the messages are sent and received by the intended parties, thereby preventing impersonation attacks.

\textbf{Anonymity:}
We provide a service mode, which is called \textit{private mode}, where the servers simply handle the requests from the users without considering any users' identity information. In general, anonymity implies that the server cannot identify who it is interacting with in \textit{private mode} beyond the information derived from the requests. To formalise this property, consider two users with identities $ID_0$ and $ID_1$, who interact with the server using the same requests. In \textit{private mode}, the server's view includes $\{ID_0,ID_1,inter\_info,lience(ID_0), Q, A\}$ and $\{ID_0, ID_1, inter\_info,  lience\\(ID_1), Q, A\}$, where the $lience$ represents information regarding the authenticated user, and $inter\_ info$ represents the information generated during the interaction between the server and users $ID_0$ and $ID_1$, including $m'_0$, $m'_1$,$\sigma'_0$,$\sigma'_1$, and so on generated during the blind signature process. Here, $Q$ and $A$ represent the request and response, respectively. For simplicity, we refer to these views as $view_0$ and $view_1$. Now, a challenger $C$ selects a random bit $b \in \{0,1\}$. The adversary, given $\{view_0, view_1\}$, outputs a bit $b'$. The anonymity property holds if the probability $|Pr[b' = b] - 1/2|$ is less than $1/\textit{poly}$. It is important to note that anonymity implies unlinkability, implying the server cannot determine whether two queries in \textit{private mode} originate from the same user beyond the information obtained from the queries themselves.

\textbf{Interest Protection:}
This framework is designed to protect the interests of both users and servers. Upon payment or subscription, users gain access to the services they are entitled to. In addition, users are limited to using only the services they have been granted, preventing overuse of the server's resources.

\subsection{Challenges}
To make the framework practical, several issues must be addressed. First, this framework is designed to support real-world business models, requiring compatibility with common LLM service models, as discussed in Section \ref{preliminaries:llm_services_models}. In a subscription-based mode, the framework must ensure user anonymity while also restricting access to the service for the duration of the user's subscription period. Conversely, in an API-based mode, the framework must ensure that service providers can enforce limits on the number of input tokens submitted by anonymous users, in accordance with the purchased quota. Second, the framework should be compatible with other modules within the system, making it adaptable to both existing and future business logic. In addition, because the services are primarily online, users wait for responses after submitting requests. The framework should introduce reasonable overhead to maintain security while minimising any degradation to the user experience.

\section{Proposed Framework}

In this section, we present a framework that uses partially blind signatures to ensure user anonymity within LLM services. The framework incorporates minor adjustments to the charging models and uses the partially blind signature to achieve anonymity across two prevalent business service models while ensuring compatibility with existing systems. In addition, we note that PKI is used as the default mechanism to establish basic secure communications between clients and service providers.

To achieve the dual objectives of anonymity and compatibility, our framework ensures that requests remain unlinkable to their senders while minimising modifications to the underlying LLM systems. Blind signatures were selected as the foundation of our system because they provide strong anonymity guarantees, ensuring that service providers are unable to discern the content they sign. Furthermore, because the signature serves as auxiliary data for validating requests at the outermost layer of the service, no modifications to the core LLM architectures were required, thereby ensuring compatibility.

In contemporary business models, it is common for users and service providers to negotiate certain information, such as service agreements, quality of service, and membership details. However, standard blind signatures allow users to control all the information, rendering service providers unaware of the content they sign. To address this, we opted for partially blind signatures, which allow both users and service providers (acting as signers) to jointly control certain parts of the data to be signed, effectively addressing the business requirements mentioned earlier.

To integrate partially blind signatures into practical business scenarios, we devised strategies tailored to different service models. In subscription-based models, users and service providers agree on a subscription period, which is included in the negotiated public data of users and service providers. A unique subscription period could potentially distinguish users, thereby compromising anonymity. To address this issue, our strategy records only the subscription deadline, which is standardised to specific dates (e.g., the first day of each month for monthly subscriptions). By having a large number of users share the same subscription deadlines, it becomes challenging for service providers to trace requests back to individual users. Furthermore, in this model, service providers are not required to generate multiple blind signatures at once. Instead, they issue a new blind signature after each completed query, ensuring that each legitimate user maintains only one blind signature during the valid subscription period. Service providers do not need to track which signatures have been used, unless they wish to implement a feature that prevents multiple users from using the same account simultaneously. In such cases, tracking used signatures would be necessary. These strategies significantly reduce the problem of storing a large number of blind signatures. In API-based pricing models, users and service providers must be aware of the number of tokens purchased. Directly recording token amounts could expose user identities. To address this issue, our framework charges users based on the number of requests they intend to send, with a fixed maximum number of tokens per request. This method obscures individual identities by having many users adhere to the same charging model. In addition, both users and service providers must be aware of the number of tokens purchased and the type of API being accessed, such as GPT-4 or Code Interpreter. Directly recording token amounts and API types could also expose user identities. To address this issue, our framework uses partially blind signatures, where certain information, such as the API type, is included in the public negotiated data. This method maintains user anonymity while allowing necessary information to be recorded.

The design details of our framework are as follows:

\subsection{Compatible Service Layer}
Our framework modifies only the inference stage of the LLMs' processes. Specifically, it functions as an additional layer of the LLM services, without requiring changes to other components of the models. Before sending their requests, users must submit valid tickets and corresponding unblinded signatures along with their requests to the service providers. Upon receiving this information, the service providers first verify the legitimacy of the tickets with the signatures. If the tickets are valid, the requests proceed to the underlying modules for standard processing. Thus, our framework only introduces an extra layer to the LLM services before the normal request processing, which remains decoupled from other parts of the system, ensuring compatibility with existing systems.

\subsection{Service Models under Blind Signatures and Strategies}
Our system uses partially blind signatures to guarantee user anonymity. Users generate several tickets, which are unique within the given space. Before initiating service, the users and service providers agree on certain information. The users then blind the tickets to ensure they remain unknown to the service providers. These blinded tickets are sent to the service providers to generate blind signatures. Users subsequently unblind these signatures and submit the tickets, unblinded signatures, and request data to the service providers. Only after these signatures are validated can the requests be processed further. To accommodate practical business models, our framework provides solutions for both API- and subscription-based modes, as shown below:

\textbf{Service in API-based Mode:}
Users select pricing plans based on the number of requests they intend to send. The charging model is slightly altered in our framework; instead of charging based on tokens, users purchase a fixed number of requests, with the maximum number of tokens per request determined and fixed in the pricing plan. After selecting a pricing plan, users register with the verified service providers under a PKI system and pay for the number of requests (denoted as $N$) they intend to use. Users then generate $N$ local tickets $T=\{t_0, t_1, ..., t_N\}$. In addition, the users and service providers agree on the format for the pricing plan information $info$, which includes necessary details such as the type of LLM and maximum number of tokens per request. This information is hashed using the agreed-upon function $Ag$, producing $info'=Ag(info)$. The users then blind the tickets as $T'=blind(T, info')$ and send them to the service providers. The service providers generate the signatures $\sigma=sign(T', info', sk)$, where $sk$ denotes the service providers' secret key, and return them to the users. Note that the service provider maintains a used signature list $\sigma_{used}$ for all the users. The users then unblind the signatures to obtain $\sigma'=unblind(\sigma)=\{\sigma_1', \sigma_2', ..., \sigma_n'\}$, with each signature corresponding to a ticket in $T$. At this point, the users can log out, clear their login information, and enter anonymous mode. When sending a new request $m$, the user includes one unused signature $\sigma_i \in \sigma$ and $info$ to the service provider. Subsequently, the service provider checks whether $\sigma_i$ is in $\sigma_{used}$ and verifies $\sigma_i$ using the service provider’s public key $pk$. If $\sigma_i$ is valid and not in $\sigma_{used}$, the request is processed further and $\sigma_i$ is added to $\sigma_{used}$. The steps described are shown in Fig.~\ref{Figure:api mode interaction diagram}.
\renewcommand{\figurename}{Fig.}
\begin{figure}[!t]
	\centering
	\includegraphics [width=1\linewidth ]{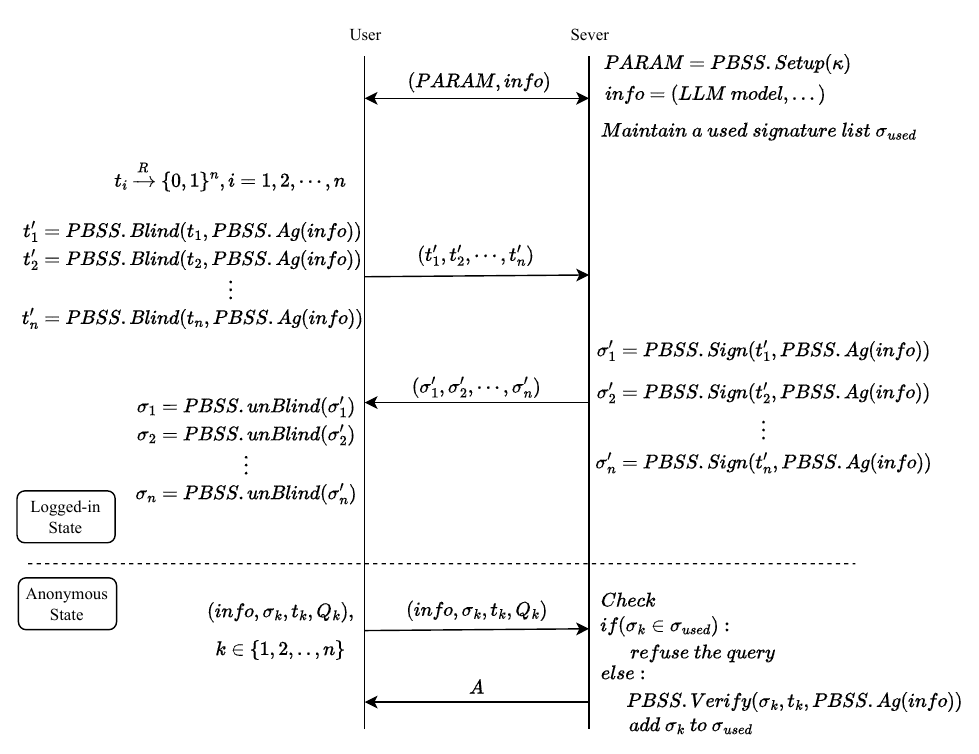}
	\caption{API-based mode interaction.}
	\label{Figure:api mode interaction diagram}
\end{figure}

\textbf{Service in Subscription-based Mode:}
In this mode, users first select a subscription plan. The charging model is also slightly modified in our framework; the service providers fix the subscription period for users, that is, the end date of each subscription period is predetermined (e.g., in a monthly subscription, the deadline is fixed on the 1st of the following month). Users register and log in to the verified service providers' system (under PKI), then choose and pay for the subscription plan. In addition, the users and service providers agree on the format for the subscription plan information $info$, which includes necessary details such as the type of LLM and subscription period deadline. This information is hashed using the agreed-upon function $Ag$, producing $info'=Ag(info)$. The users then generate a local ticket $t_0$, blind it as $t_0'=blind(t_0, info')$, and send it to the service providers. The service providers generate the signature $\sigma_0=sign(t_0', info', sk)$, where $sk$ denotes the service providers' secret key, and return it to the users. The users then unblind the signature to obtain $\sigma_0'=unblind(\sigma_0)$. At this point, the users can log out, clear their login information, and enter anonymous mode. Before sending requests, users generate a new local ticket $t_1$, blind it to $t_1'=blind(t_1, info')$, and send the request $m$, $t_0$, $\sigma_0'$ and $t_1'$ to the service provider. The service provider first checks whether $info$ is within the subscription period's deadline and then validates the signature $\sigma_0'$ using the service provider's public key $pk$. If valid, the request is processed further. Before sending the result back to the user, the service provider generates a new signature $\sigma_1=sign(t_1', info', sk)$ and returns it along with the result. The users can then use $\sigma_1$ to repeat the process for subsequent requests, provided they are submitted before the agreed-upon deadline. The steps described are shown in Fig.~\ref{Figure:sub mode interaction diagram}.
\begin{figure}[!t]
	\centering
	\includegraphics [width=1\linewidth ]{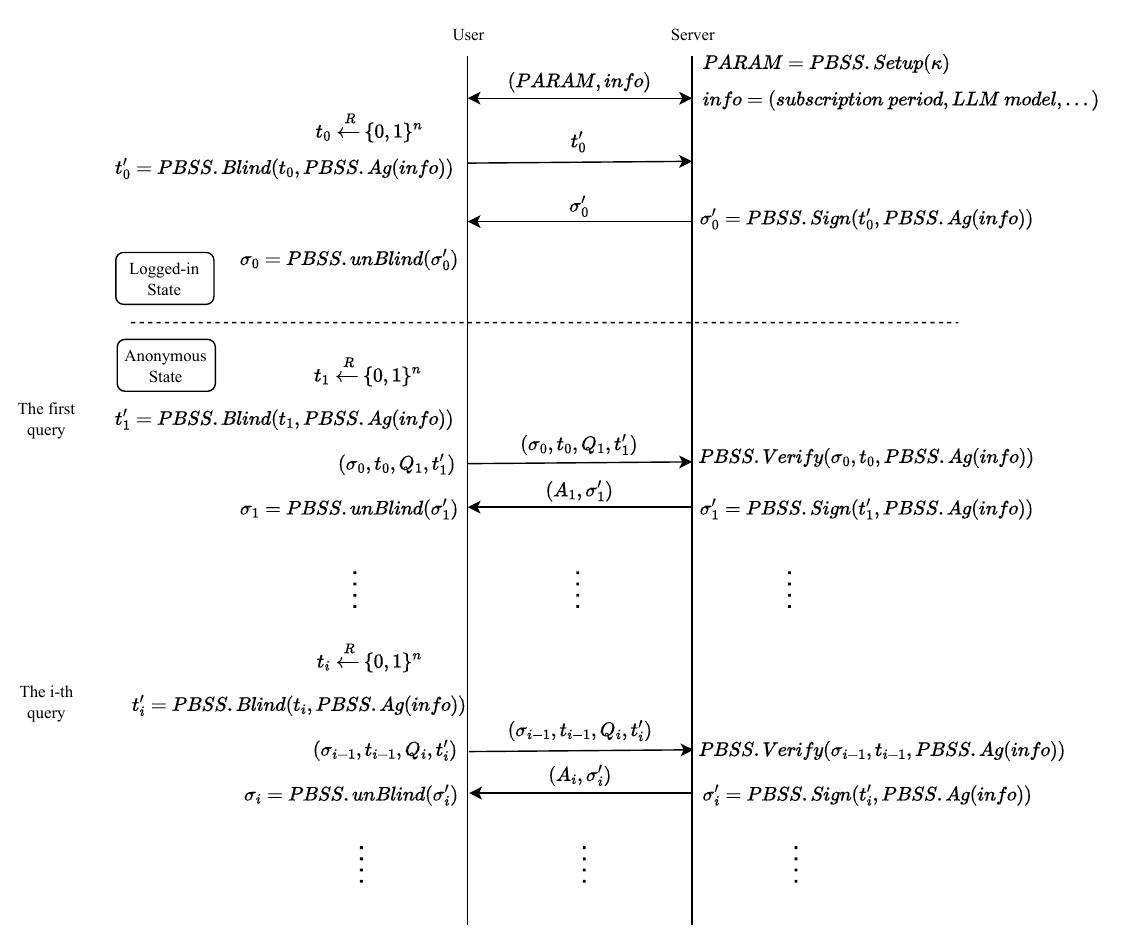}
	\caption{Subscription-based mode interaction.}
	\label{Figure:sub mode interaction diagram}
\end{figure}

\textbf{Remark:} In the subscription-based mode, the service providers do not need to maintain a list of used signatures. The information disclosed in the newly generated partially blind signatures for each round is consistent with the previous signatures; therefore, users cannot gain additional services. However, if the service providers record every used signature, they can provide an additional feature to prevent multiple users from using the same account simultaneously.

\subsection{Security Analysis}

\begin{theorem}
    The proposed framework provides secure communication.
\end{theorem}

\begin{proof}
    Under the PKI architecture, confidentiality, integrity, and server-side authentication can be ensured for communications between users and the service providers. In the login mode, the service providers can also ensure the users' identities to provide services and sign the signatures.
\end{proof}

\begin{theorem}
    The proposed framework provides anonymity.
\end{theorem}

\begin{proof}
    Assume there exists an adversary capable of breaking anonymity, specifically one that can corrupt the server and associate a query with a user's identity across multiple queries. In our scheme, which uses a partially blind signature mechanism, the server only has access to the blinded identity information. From the server's perspective, each query appears as a random sequence of numbers. However, if an adversary can infer user information from these seemingly random numbers, they would be able to link specific signatures to corresponding users. This would compromise the unlinkability property of the underlying partially blind signature scheme.
\end{proof}

\begin{figure}[!t]
	\centering
	\includegraphics [width=1\linewidth ]{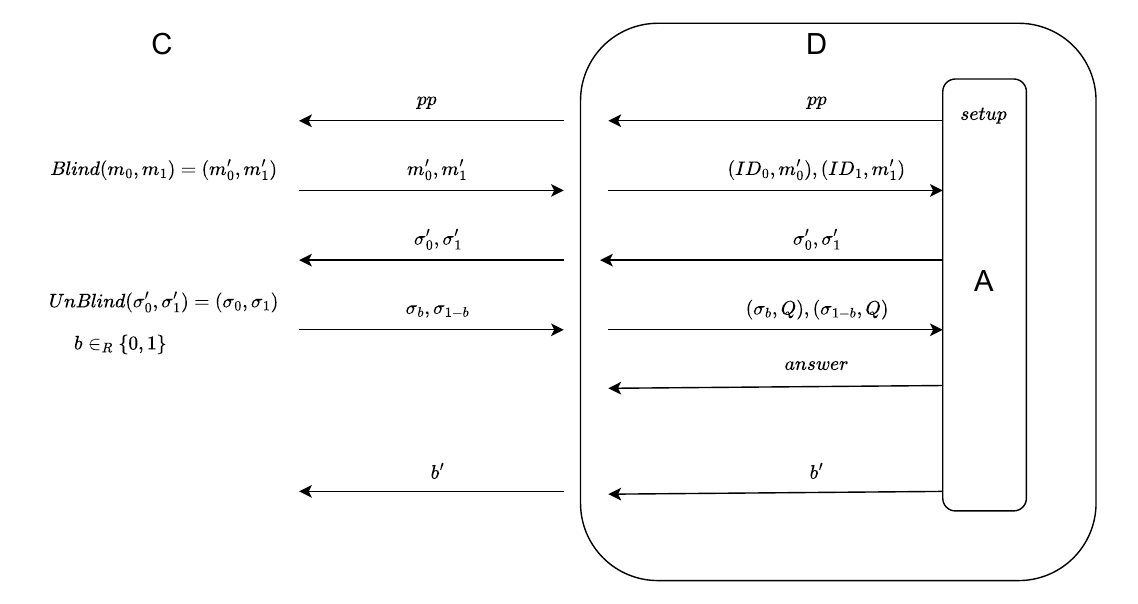}
	\caption{Security reduction of our construction from a partially blind signature.}
	\label{Figure:reduction interaction diagram}
\end{figure}

\textbf{Reduction:}
Assume there exists an adversary $A$, which is a probabilistic polynomial-time algorithm that violates anonymity. We then construct an algorithm $D$ that uses $A$ as a black-box oracle to break the unlinkability assumption of the partially blind signature scheme. Specifically, the input to $D$ is $(m'_0, m'_1, \sigma_b, \sigma_{1-b})$, and $D$ attempts to guess the bit $b'$. The reduction algorithm proceeds as shown in Fig.~\ref{Figure:reduction interaction diagram}. 
First, the challenger sends $m'_0$ and $m'_1$, which are the blinded versions of $m_0$ and $m_1$, respectively, to $D$. $D$ appends the user IDs to both $m'_0$ and $m'_1$ and forwards them to the adversary. The adversary computes the blinded signatures $\sigma'_0$ and $\sigma'_1$ and returns them to the challenger. The challenger then unblinds these signatures to obtain $\sigma_0$ and $\sigma_1$, respectively.
Subsequently, the challenger randomly selects a bit $b$, sending $\sigma_{b}$ and $\sigma_{1-b}$ to $D$. $D$ randomly selects a query and sends it to the adversary $A$, who responds with an answer and a guess for $b'$.
From the adversary's perspective, they observe two views, $view_0$ and $view_1$, and assign $\sigma_b$ and $\sigma_{1-b}$, respectively, based on these views. If the adversary can correctly determine $b = b'$ with high probability, $D$ can leverage this determination to infer the value of $b$ chosen by the challenger. This allows $D$ to solve the signature attribution problem and ultimately break the unlinkability assumption.

\textbf{Success Probability:}
The challenger's success probability is contingent on the output $b'_D$ produced by $D$, which, in turn, is influenced by the output $b'_A$ from the adversary. Therefore, the overall success probability of the challenger is directly dependent on the success probability of the adversary.

$$
\begin{aligned}
    Adv_{Unlinkability}(\lambda)&=\left|Pr\left[ b'_D=0 \& b=0 \right]-Pr\left[ b'_D=1 \& b=1 \right]\right| \\
    &=\left|Pr\left[ b'_A=0 \& b=0 \right]-Pr\left[ b'_A=1 \& b=1 \right]\right| \\
    &=\left|Pr\left[b'_A=b]-1/2\right]\right| \\
    &=Adv_{Anonymity}(\lambda)
\end{aligned}
$$
so $Adv_{Unlinkability}(\lambda)=Adv_{Anonymity}(\lambda)\leq negl(\lambda)$

\begin{theorem}
    The proposed framework provides interest protection.
\end{theorem}

\begin{proof}
    First, in the anonymous mode, if an adversary disguises themselves as a legitimate user to gain access to the server's services, they would be required to forge a legitimate partial blind signature. Similarly, if users corrupted by the adversary wish to obtain additional services, they have to tamper with the public expiration date of the signature in the subscription-based mode or generate new signatures in the API-based mode. Both scenarios would compromise the unforgeability of partially blind signatures. Second, the non-repudiation of signatures guarantees the users’ rights after purchasing services. Third, the feature of partially blind signatures that include public information clearly specify the types of services, maximum number of tokens per request, and subscription deadline of services. Therefore, the proposed framework protects the interests of all parties involved.
\end{proof}

\section{Experiment and Performance Analysis}
\subsection{Experimental Setup}

\textbf{Implementation:}
Our implementation consisted of two components: the server and client, designed to closely mimic a real-world scenario. We used \textit{mkcert} to generate certificates for CA and server and install the root certificate on the client. The server provided HTTPS APIs to handle incoming requests. Because no modifications were required to the LLMs within our system, we used pre-trained DeepSeek-R1-Distill-Qwen-14B and DeepSeek-R1-Distill-Llama-8B models from Hugging Face to deliver the underlying inference capability on the server side. The client component sends requests to the server's API after verifying the server's identity using the root certificate. In our proposed solution, we used a standardised RSA-based partially blind signature scheme as the foundational cryptographic component; detailed specifications can be found in the study \cite{amjad-cfrg-partially-blind-rsa-03}. The modulus $n$ was configured to be 2048 bits. The two hash functions used, designated as $H_M$ and $H_{MD}$, were derived from modifications of the SHA-384 algorithm.

We conducted four sets of experiments. The first three were comparative experiments: one control group using the original LLM without partially blind signatures and two experimental groups incorporating partially blind signatures. The experimental groups were further divided into the subscription and API mode groups. In the control group, we asked five different questions(Q1-Q5)\footnote{Q1: ‘What is the tallest mountain in the world?’; Q2: ‘What does the Planck constant represent in physics?’; Q3: ‘Which planet is known as the “Red Planet”?’; Q4: ‘What is the chemical symbol for gold?’; Q5: ‘How many bones are there in an adult human body?’}, and repeated these five questions in sequence to enhance the reliability of the experiment, measuring the processing time from request to response, as well as the size of the request and response data. The same setup was applied to the experimental groups. The data from both experimental groups were compared with the control group, and the detailed results are listed in TABLE \ref{tab:quantitative_results}. In the final experiment, we focused on the API-based mode, where we generated 10 signatures in a single request. We measured this duration of the process, along with the sizes of both the request and response data. 

In addition, considering that real-world LLM services typically operate in multi-user environments and that models deployed by different organisations may adopt distinct internal mechanisms, we further conducted experiments to better reflect practical settings. Specifically, we introduced concurrent-user experiments with varying numbers of users, as well as cross-provider comparison experiments. These evaluations aim to verify that our framework introduces no additional overhead even under multi-user concurrency and across heterogeneous model providers.We designed the following experiments. Experiment 1 was conducted on the deepseek-ai/DeepSeek-R1-Distill-Qwen-14B model. Under three modes (api, sub, and normal)we measured the performance differences when 5 and 10 users concurrently queried the system. In each setting, the users repeatedly issued five long questions (LQ1-–LQ5) in a loop for five rounds. Experiment 2 followed the same experimental design but was conducted on the Qwen/Qwen2.5-14B model. These experiments collectively demonstrate that the proposed framework maintains negligible overhead under concurrent usage and remains robust across models from different institutions.

Availability. We provide a reproducible implementation in \href{https://gitee.com/lpacino/llm_-secure.git}{Gitee}

\textbf{Hardware:}

The experimental setup included a server equipped with an Intel Xeon (Icelake) Platinum processor operating at 2.9 GHz, with maximum turbo frequency reaching 3.5 GHz and supported by 128 GB of RAM. The GPU is NVIDIA A10 with 24 GB of video memory $\times 2$. The GPU and CPU support a PCIe 4.0 connection. The client system used an Intel® Core™ i7-9750H processor featuring six cores running at 2.6 GHz, complemented by 8 GB of RAM. The server and client components were deployed on a commercial cloud instance and a standard residential machine, respectively. All experiments were performed over a WAN to accurately replicate real-world deployment scenarios.
\subsection{Experimental Results and Performance Analysis}

\textbf{Quantitative Results:}

To evaluate the performance of our proposed scheme, we conducted a series of experiments measuring communication time, request size, and response size. We compared the baseline system without blind signatures (normal) against our two proposed modes: the subscription-based blind signature mode (subscription) and the API-based blind signature mode (API).The following analysis is based on the results of the DeepSeek-R1-Distill-Qwen-14B model. The results of the DeepSeek-R1-Distill-Llama-8B model are provided in the appendix.

The aggregated results, including the mean ($\mu$) and standard deviation ($\sigma$) for each metric across five different questions (Q1-Q5), are listed in Table \ref{tab:quantitative_results}. For a more intuitive visualisation of the data distribution and variance, Figs.~\ref{Figure:Q1 - Q5 Communication Time Comparison(14B)},\ref{Figure:Q1 - Q5 Request Size Comparison(14B)}, and \ref{Figure:Q1 - Q5 Response Size Comparison(14B)} show the performance of the three modes using box plots for each metric.

As can be observed from the data, the introduction of our partial blind signature mechanism led to a noticeable increase in data payload, particularly for the request size in both the Subscription and API modes and for the response size in the Subscription mode. In contrast, the overall impact on end-to-end communication time appeared to be minimal, with no significant order-of-magnitude changes observed across all tested questions when compared to the baseline.

\begin{table*}[ht]
\centering
\caption{Mean ($\mu$) and standard deviation ($\sigma$) of performance metrics across all modes.}
\label{tab:quantitative_results}
\resizebox{\textwidth}{!}{%
\begin{tabular}{c|l|cc|cc|cc}
\hline
\multirow{2}{*}{\textbf{Question}} & \multirow{2}{*}{\textbf{Metric}} & \multicolumn{2}{c|}{\textbf{Normal Mode}} & \multicolumn{2}{c|}{\textbf{Subscription Mode}} & \multicolumn{2}{c}{\textbf{API Mode}} \\
\cline{3-8}
& & \textbf{$\mu$} & \textbf{$\sigma$} & \textbf{$\mu$} & \textbf{$\sigma$} & \textbf{$\mu$} & \textbf{$\sigma$} \\
\hline
\multirow{3}{*}{\textbf{Q1}}
& Time (s)           & 37.7 & 0.3 & 37.3 & 0.7 & 36.5 & 0.7 \\
& Request Size (B)   & 652.0 & 0.0 & 1767.8 & 13.0 & 1321.4 & 0.5 \\
& Response Size (B)  & 169.0 & 0.0 & 530.0 & 0.0 & 169.0 & 0.0 \\
\hline
\multirow{3}{*}{\textbf{Q2}}
& Time (s)           & 321.2 & 4.4 & 323.1 & 3.8 & 319.6 & 2.2 \\
& Request Size (B)   & 661.0 & 0.0 & 1781.2 & 1.5 & 1330.4 & 0.5 \\
& Response Size (B)  & 1273.0 & 0.0 & 1634.9 & 2.7 & 1273.0 & 0.0 \\
\hline
\multirow{3}{*}{\textbf{Q3}}
& Time (s)           & 321.7 & 3.3 & 322.0 & 2.5 & 320.9 & 6.3 \\
& Request Size (B)   & 652.0 & 0.0 & 1772.7 & 0.5 & 1321.4 & 0.5 \\
& Response Size (B)  & 1197.0 & 0.0 & 1558.0 & 0.0 & 1197.0 & 0.0 \\
\hline
\multirow{3}{*}{\textbf{Q4}}
& Time (s)           & 290.1 & 3.4 & 292.6 & 4.9 & 299.3 & 5.2 \\
& Request Size (B)   & 647.0 & 0.0 & 1767.6 & 0.5 & 1316.4 & 0.5 \\
& Response Size (B)  & 1078.0 & 0.0 & 1439.0 & 0.0 & 1078.0 & 0.0 \\
\hline
\multirow{3}{*}{\textbf{Q5}}
& Time (s)           & 150.9 & 2.8 & 152.0 & 1.9 & 153.3 & 2.2 \\
& Request Size (B)   & 658.0 & 0.0 & 1778.8 & 0.4 & 1327.4 & 0.5 \\
& Response Size (B)  & 588.0 & 0.0 & 949.0 & 0.0 & 588.0 & 0.0 \\
\hline
\end{tabular}%
}
\end{table*}

In addition, the processing time for generating 10 signatures in the API-based mode was measured at 0.78 s, involving a request size of 3676 B and yielding a response size of 7019 B.

\begin{figure}[!t]
	\centering
	\includegraphics [width=1\linewidth ]{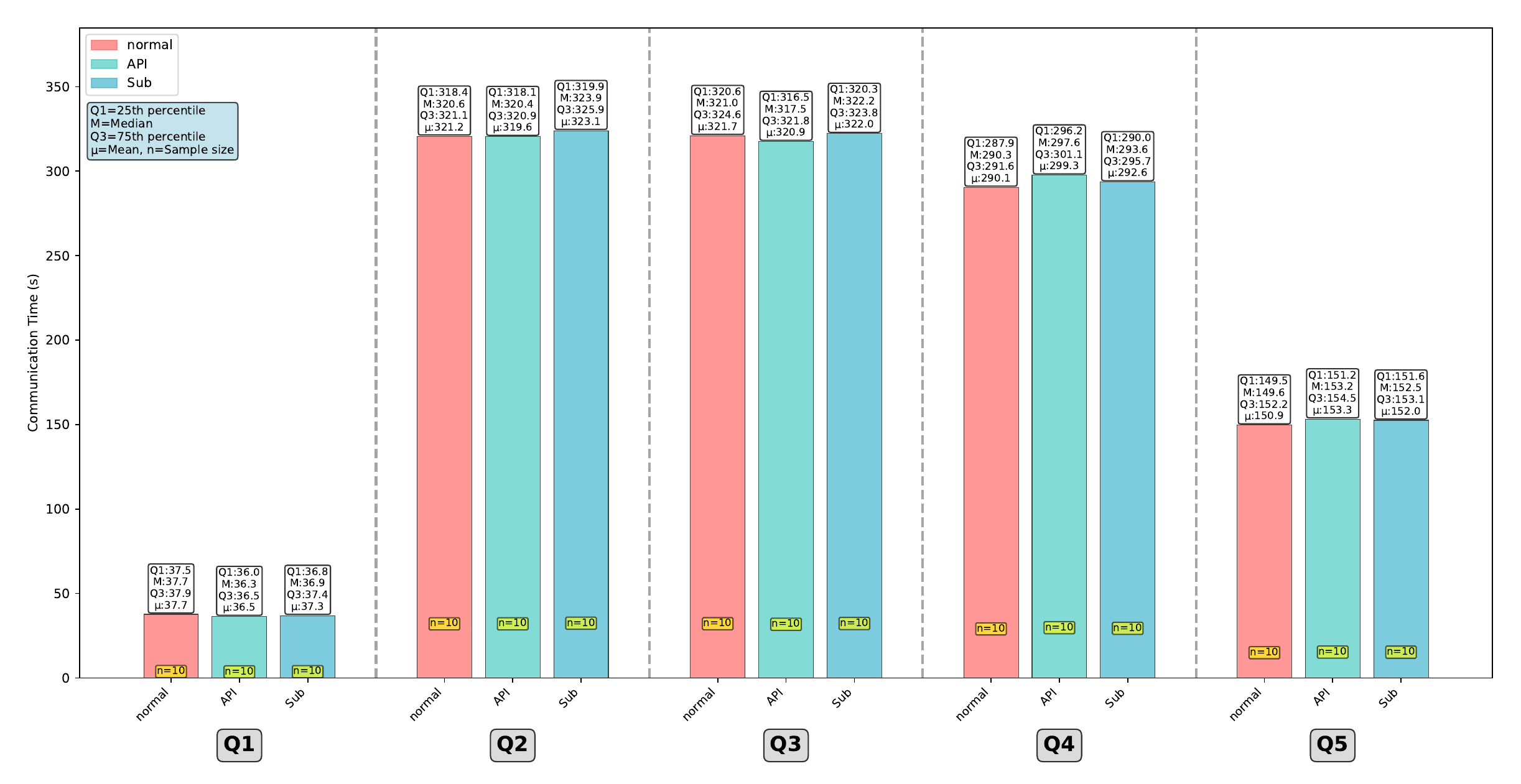}
	\caption{Q1 - Q5 communication time momparison.}
	\label{Figure:Q1 - Q5 Communication Time Comparison(14B)}
\end{figure}

\begin{figure}[!t]
	\centering
	\includegraphics [width=1\linewidth ]{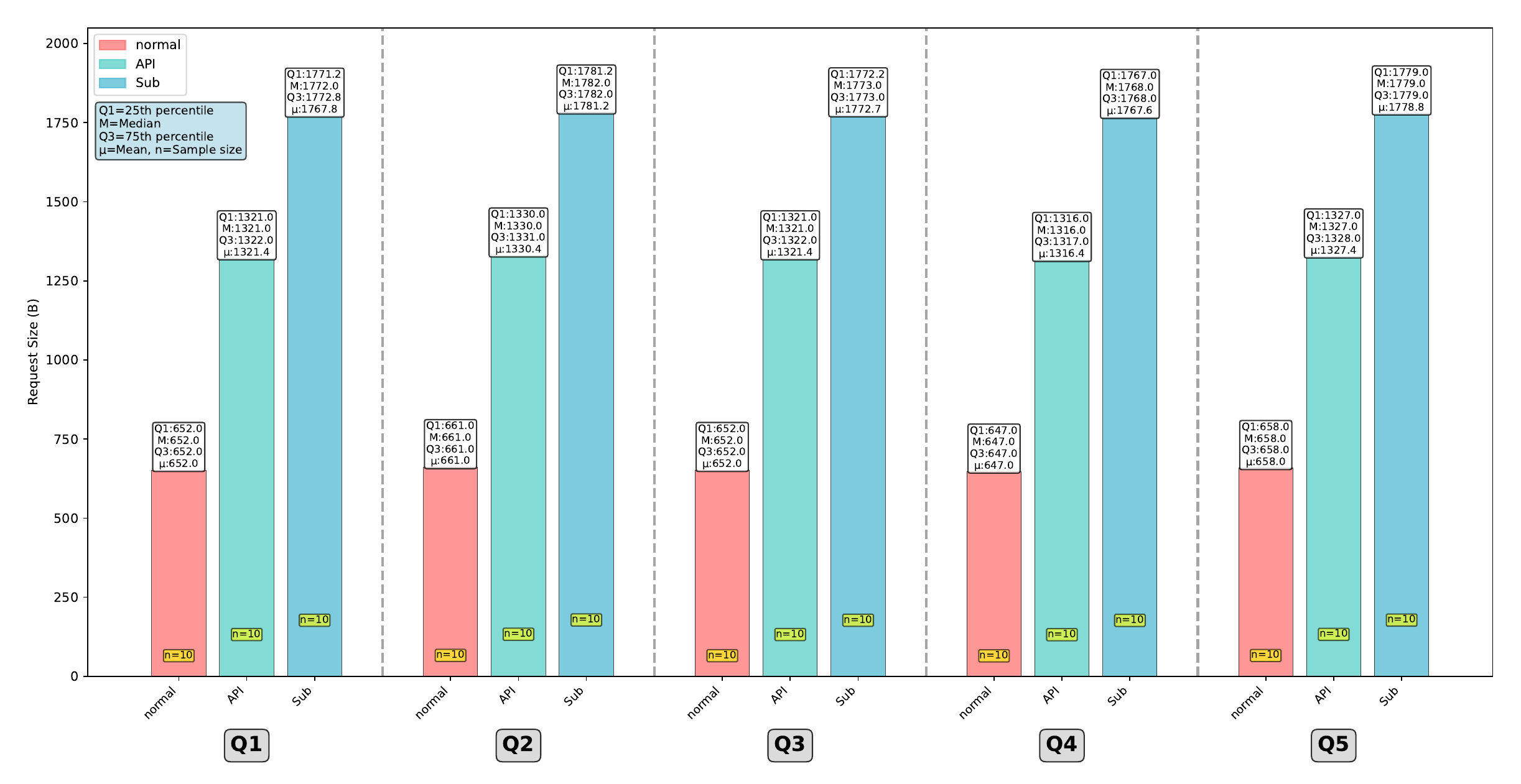}
	\caption{Q1 - Q5 communication request size comparison.}
	\label{Figure:Q1 - Q5 Request Size Comparison(14B)}
\end{figure}

\begin{figure}[!t]
	\centering
	\includegraphics [width=1\linewidth ]{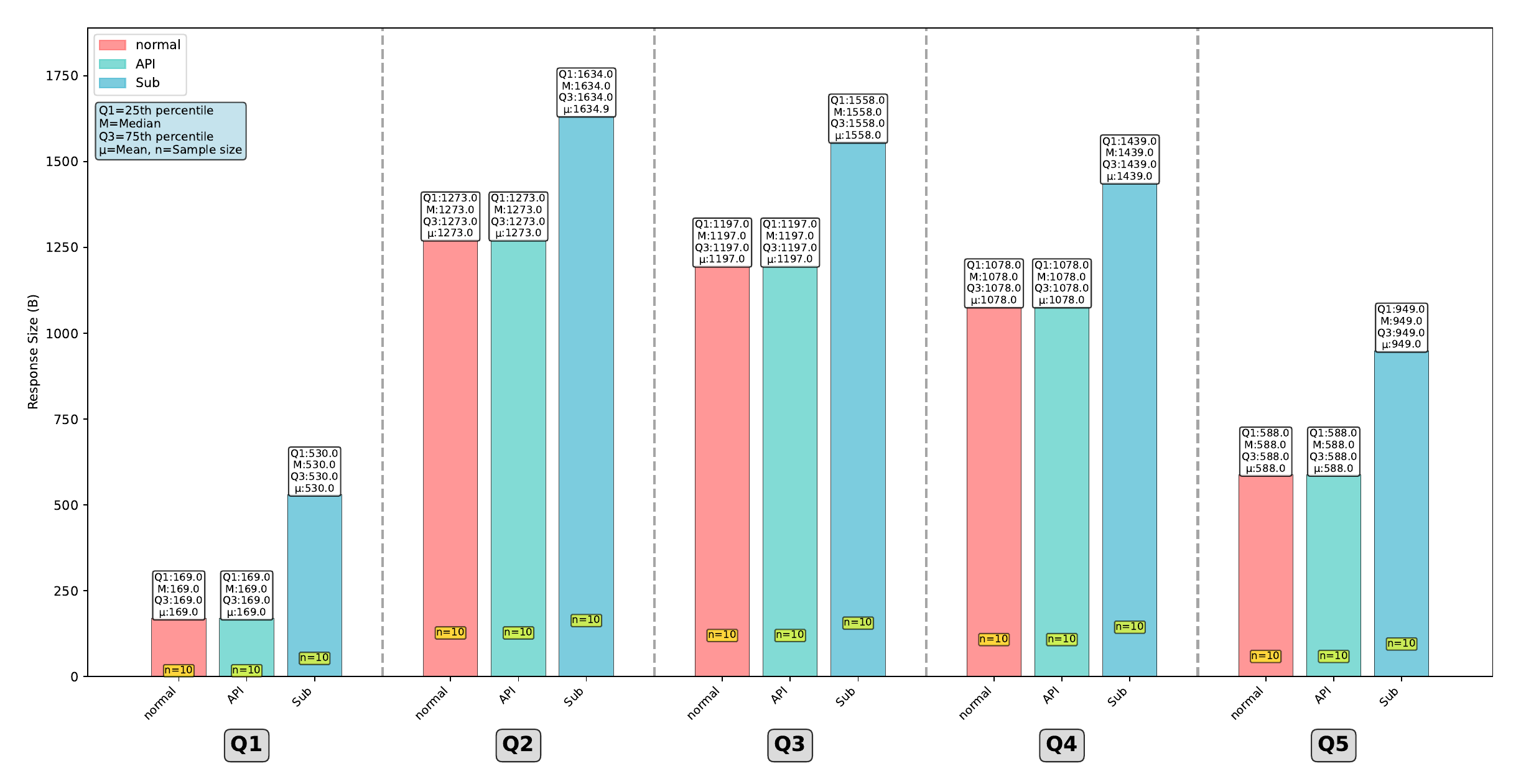}
	\caption{Q1 - Q5 communication response size comparison.}
	\label{Figure:Q1 - Q5 Response Size Comparison(14B)}
\end{figure}

Simultaneously, to reflect the stability of the differences between our framework and the normal framework, we calculated the differences between the data of the API mode, SUB mode and normal mode, and then drew the box plot in Figs.~\ref{Figure:Q1 - Q5 Communication Time differentials Comparison(14B)},\ref{Figure:Q1 - Q5 Request Size differentials Comparison(14B)},\ref{Figure:Q1 - Q5 Response Size differentials Comparison(14B)} based on these differences.

To further demonstrate that our framework does not introduce exponential overhead as the number of concurrent users increases, and that it remains applicable across different model providers, we additionally conducted a set of concurrency experiments. As shown in Figs.~\ref{Figure:Different Numbers of Users Communication Time Comparison(Deepseek 14B)},\ref{Figure:Different Numbers of Users Request Size Comparison(Deepseek 14B)},\ref{Figure:Different Numbers of Users Response Size Comparison(Deepseek 14B)}, we report the average overhead of the DeepSeek-R1-Distill-Qwen-14B model under three modes (api, sub, and normal) with 5 and 10 concurrent users. In each setting, the users repeatedly queried the system with five long questions LQ1-–LQ5 (in \ref{app2}) over five consecutive rounds. The results confirm that the overhead remains stable as concurrency increases and that the proposed framework maintains consistent performance across varying system loads.

\begin{figure}[!t]
	\centering
	\includegraphics [width=1\linewidth ]{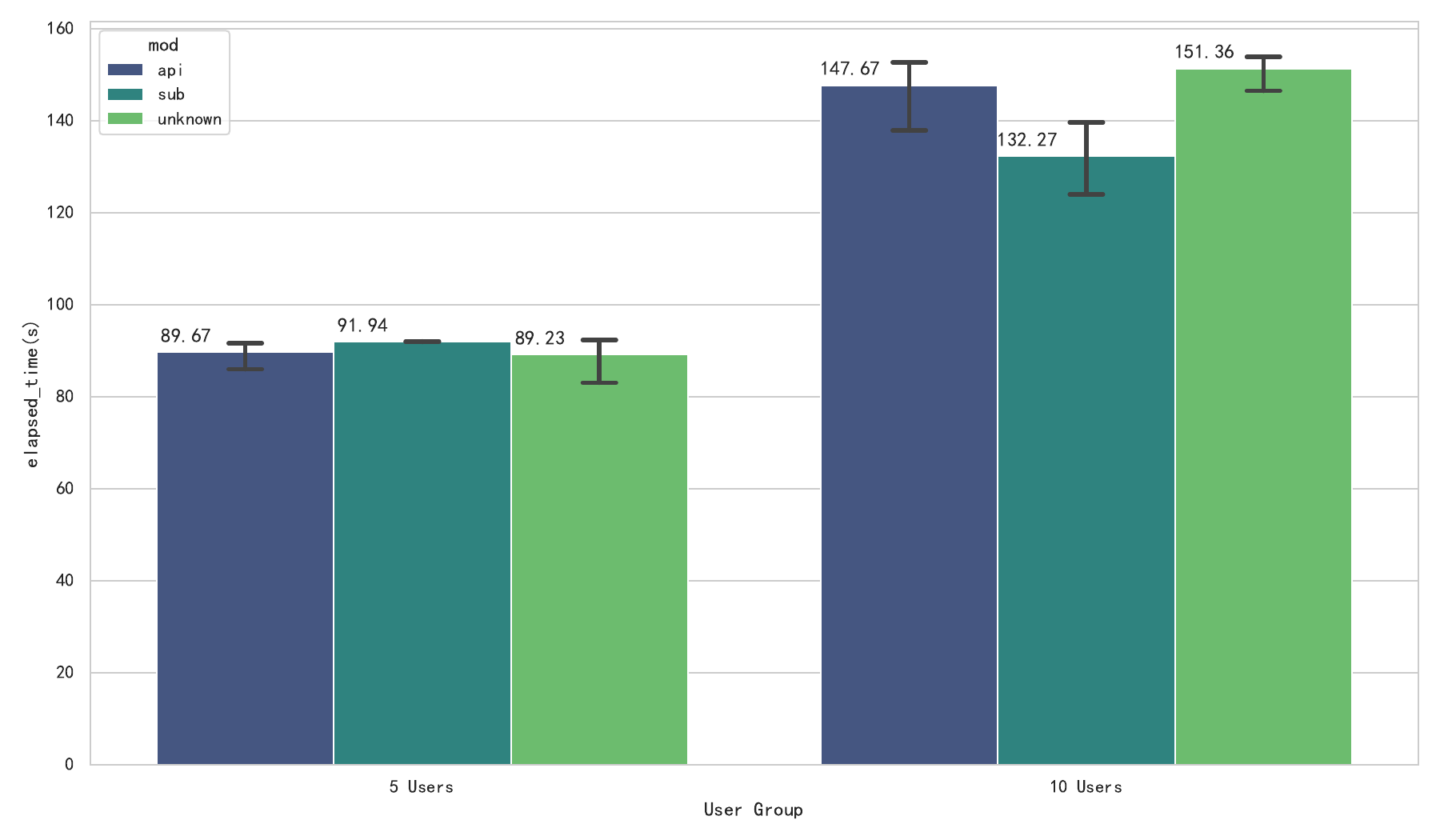}
	\caption{Different Numbers of Users Communication Time Comparison(Deepseek 14B).}
	\label{Figure:Different Numbers of Users Communication Time Comparison(Deepseek 14B)}
\end{figure}

\begin{figure}[!t]
	\centering
	\includegraphics [width=1\linewidth ]{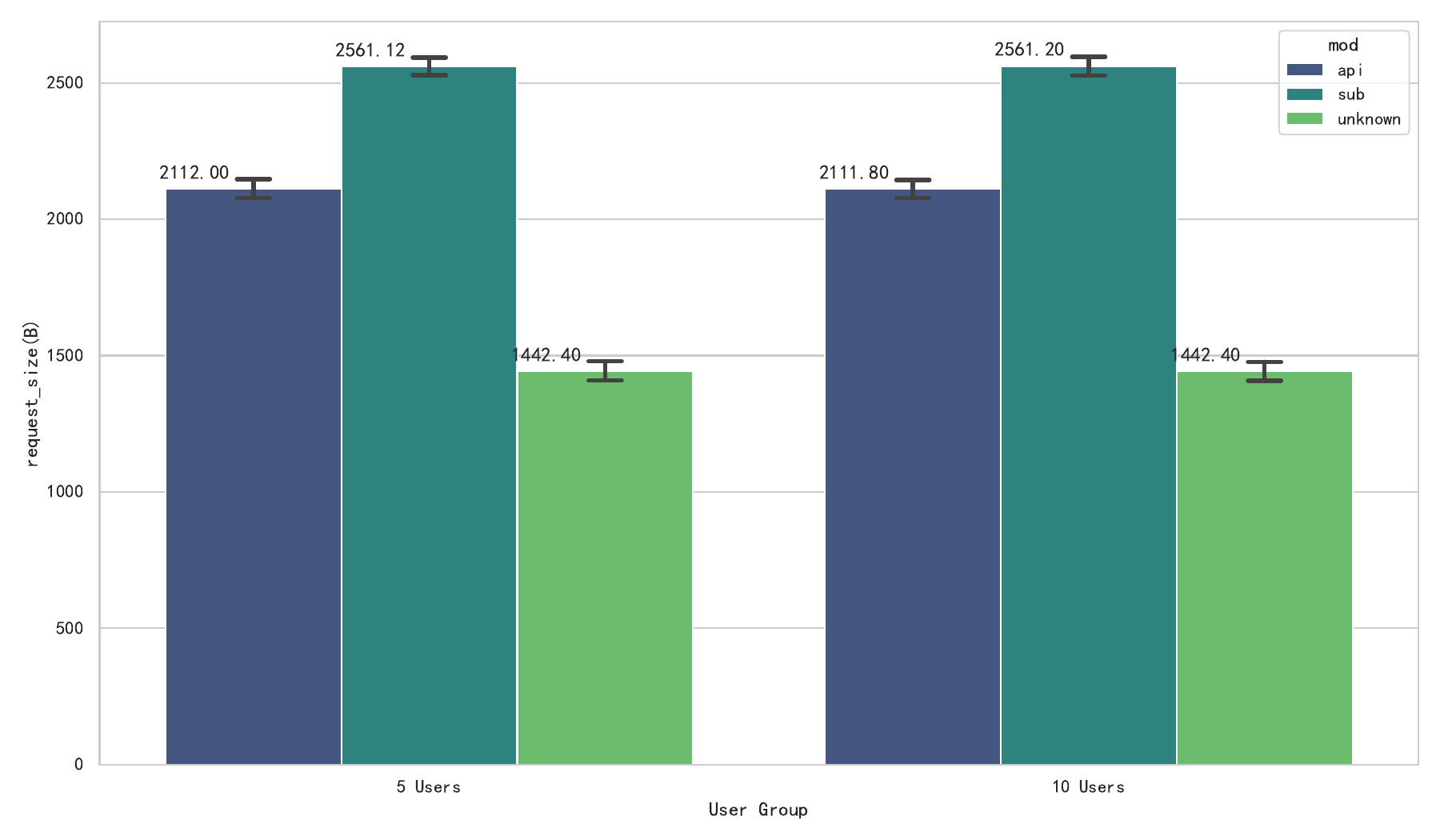}
	\caption{Different Numbers of Users Request Size Comparison(Deepseek 14B).}
	\label{Figure:Different Numbers of Users Request Size Comparison(Deepseek 14B)}
\end{figure}

\begin{figure}[!t]
	\centering
	\includegraphics [width=1\linewidth ]{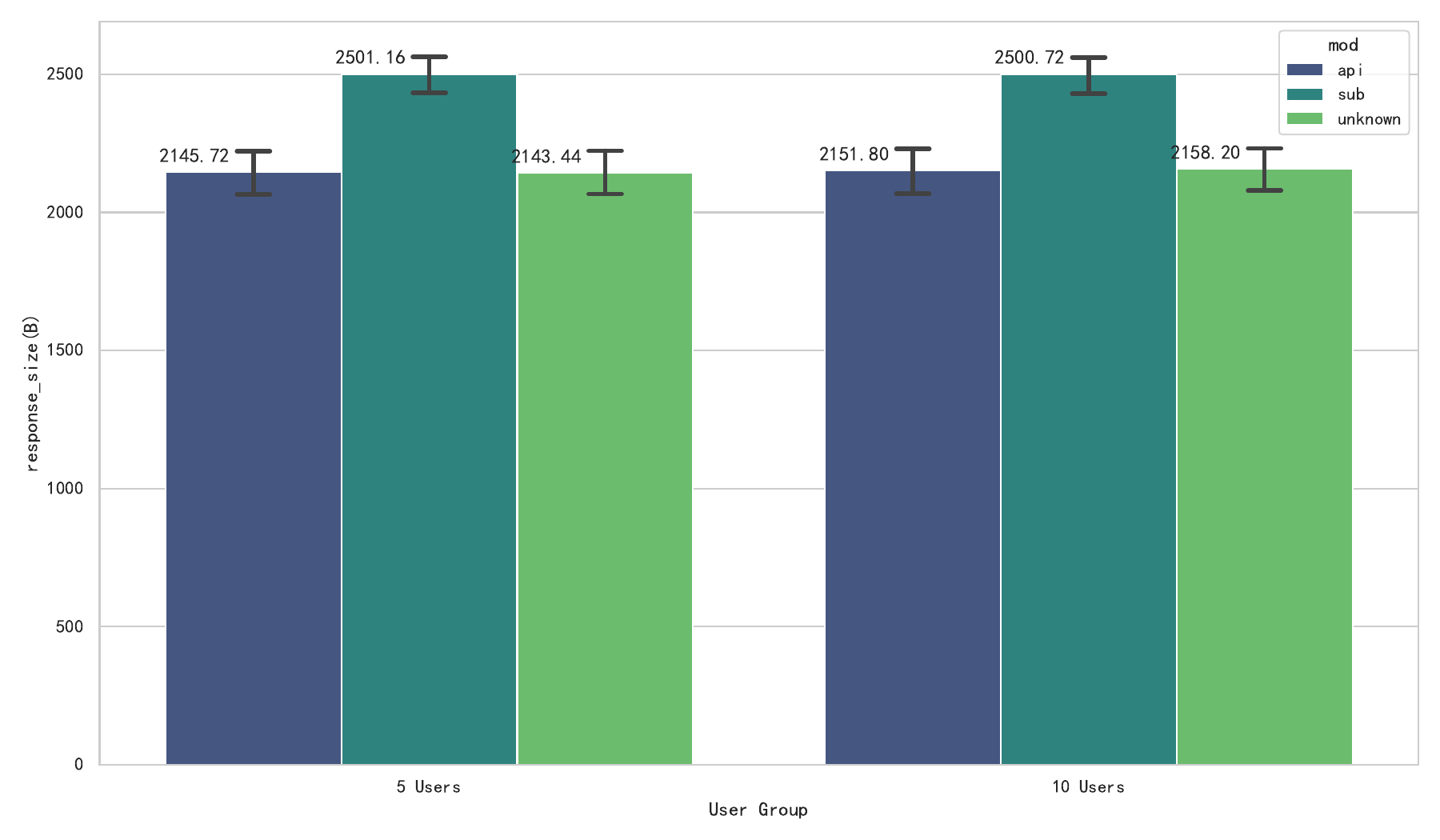}
	\caption{Different Numbers of Users Response Size Comparison(Deepseek 14B).}
	\label{Figure:Different Numbers of Users Response Size Comparison(Deepseek 14B)}
\end{figure}

Similarly, in Figs.~\ref{Figure:Different Numbers of Users Communication Time Comparison(Qwen2.5 14B)},\ref{Figure:Different Numbers of Users Request Size Comparison(Qwen2.5 14B)},\ref{Figure:Different Numbers of Users Response Size Comparison(Qwen2.5 14B)}, we present the average overhead of the Qwen2.5-14B model under the three modes (api, sub, and normal) with 5 and 10 concurrent users. In each configuration, the users issued five long questions LQ1–-LQ5(in \ref{app2}) in five successive rounds.

\begin{figure}[!t]
	\centering
	\includegraphics [width=1\linewidth ]{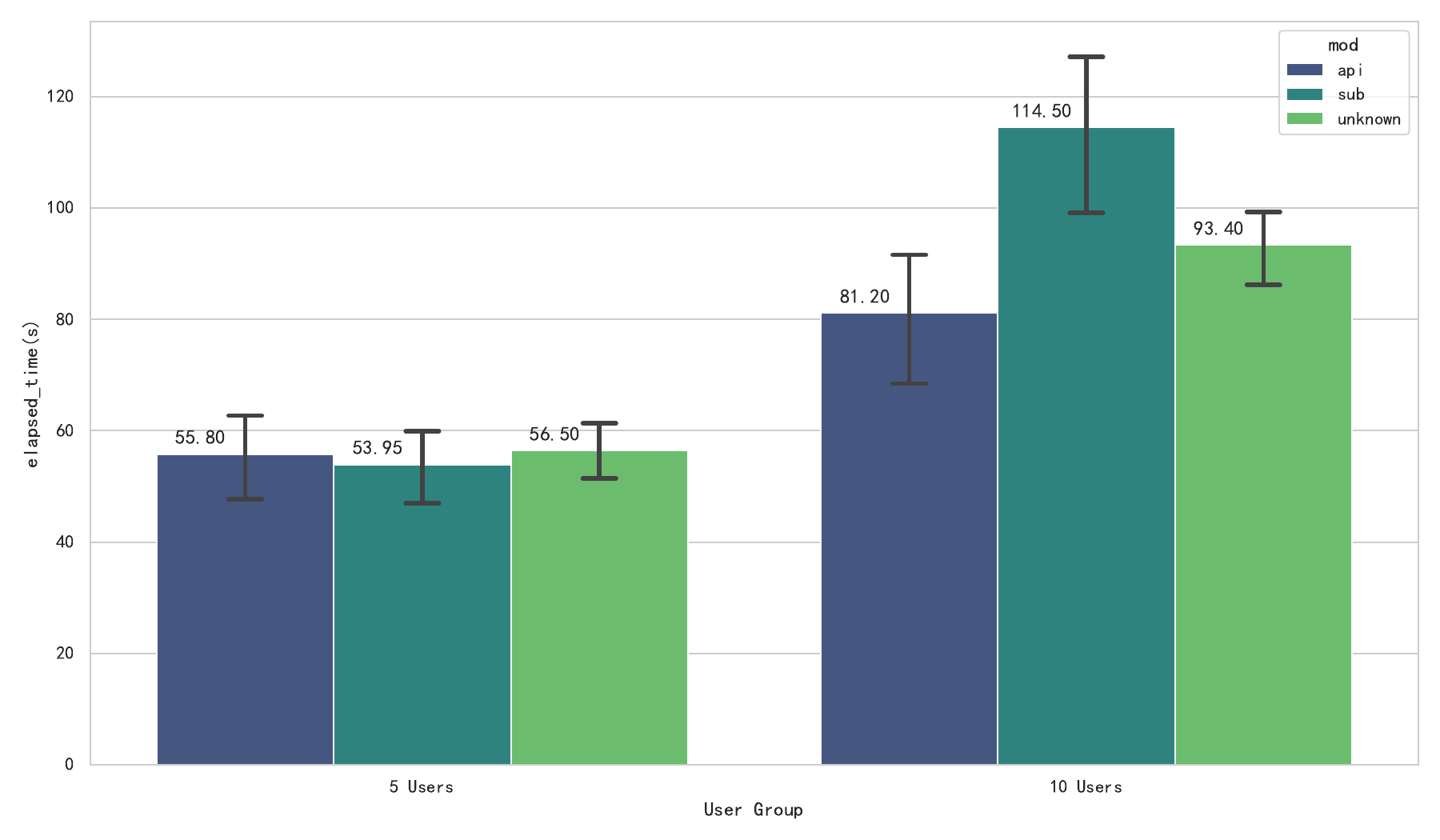}
	\caption{Different Numbers of Users Communication Time Comparison(Qwen2.5 14B).}
	\label{Figure:Different Numbers of Users Communication Time Comparison(Qwen2.5 14B)}
\end{figure}

\begin{figure}[!t]
	\centering
	\includegraphics [width=1\linewidth ]{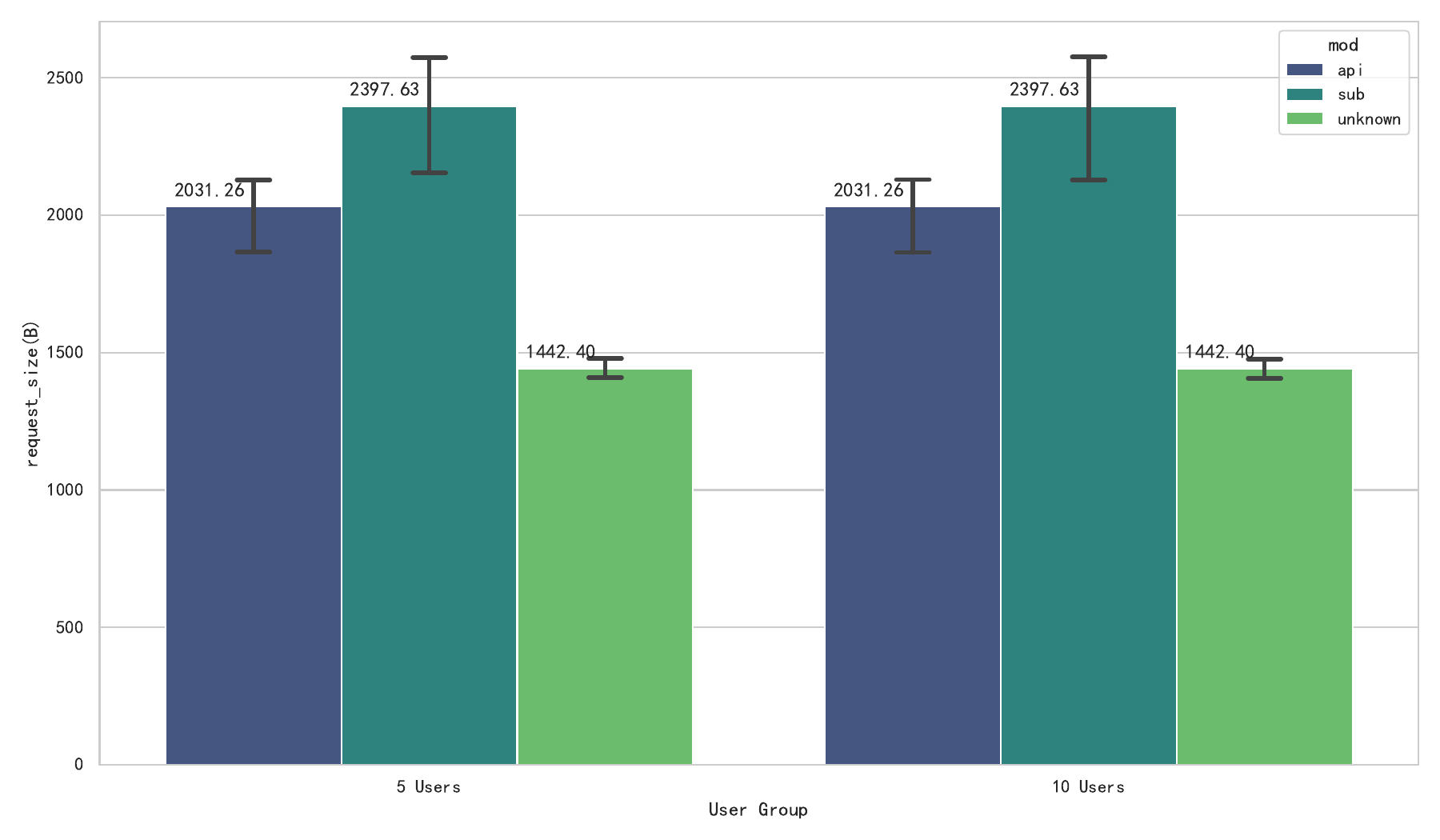}
	\caption{Different Numbers of Users Request Size Comparison(Qwen2.5 14B).}
	\label{Figure:Different Numbers of Users Request Size Comparison(Qwen2.5 14B)}
\end{figure}

\begin{figure}[!t]
	\centering
	\includegraphics [width=1\linewidth ]{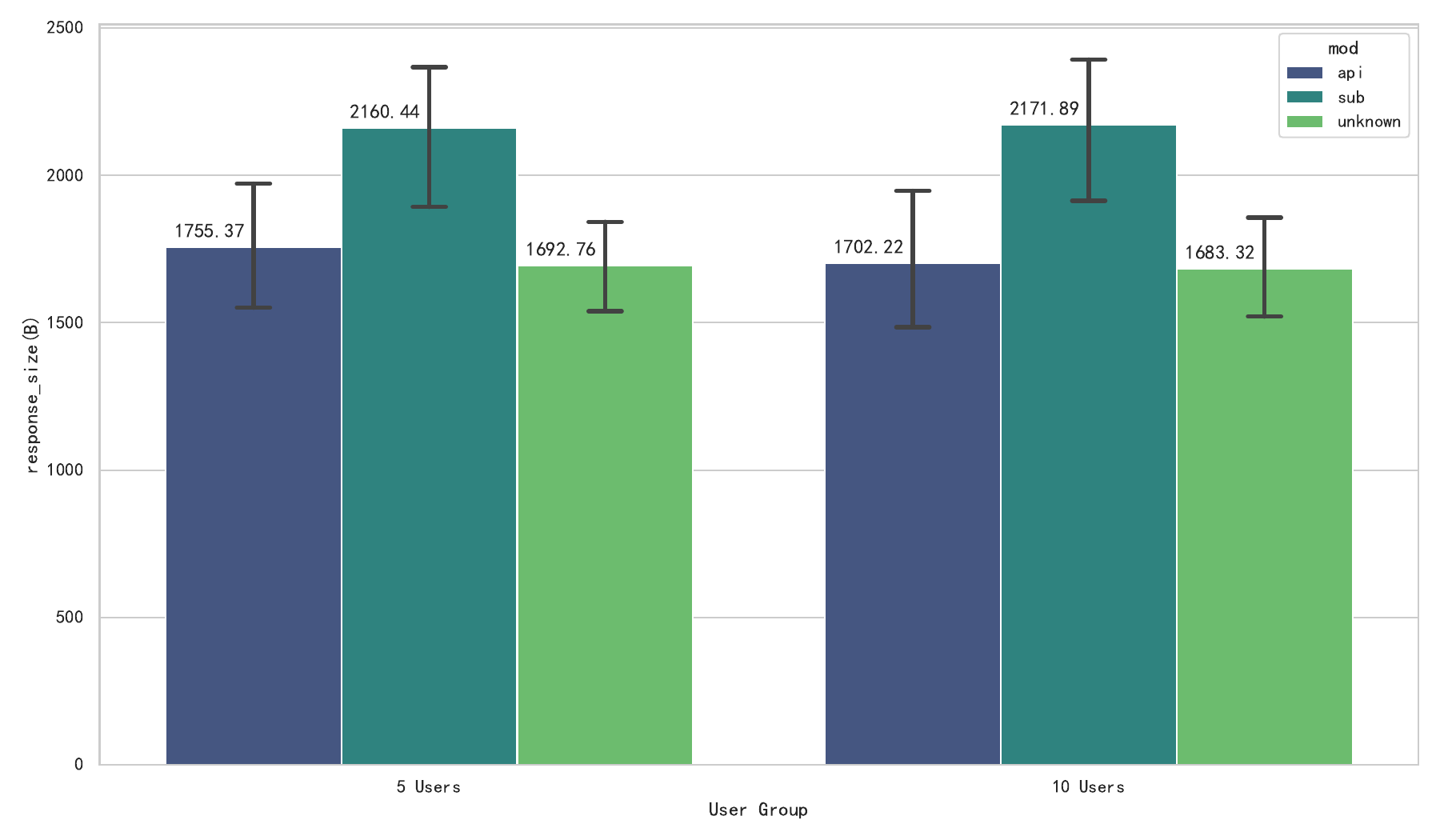}
	\caption{Different Numbers of Users Response Size Comparison(Qwen2.5 14B).}
	\label{Figure:Different Numbers of Users Response Size Comparison(Qwen2.5 14B)}
\end{figure}


\textbf{Computation and Communication Overhead Analysis:}

In this section, we examine the computation and communication overhead introduced by our framework. We analyse the performance differences between our proposed modes and the baseline to quantify the precise cost of implementing the security layer.

Computation Overhead. The computational overhead is primarily composed of the cryptographic operations required for generating, blinding, signing, and verifying the messages. Fig.~\ref{Figure:Q1 - Q5 Communication Time differentials Comparison(14B)} shows the distribution of the additional time incurred by the Subscription and API modes compared to the baseline. The mean time difference was consistently low across all experiments. The results showed that the impact on performance was minimal, though its statistical significance varied depending on the query's characteristics. For instance, for Q1, a two-sample t-test revealed a statistically significant time difference for the API mode compared to the baseline (p = 0.0007). However, the absolute mean overhead was small, at only 1.2 s. This indicates that while the overhead was consistent enough to be statistically detectable, it was practically negligible in magnitude. In contrast, for Q2, the time differences were identified to be not statistically significant for either the API mode (p = 0.3600) or the Subscription mode (p = 0.3264). This suggests that for this query, the minor fluctuations in response time were well within the range of normal system variability, and no systematic overhead could be attributed to our scheme. Considered together, these results confirmed that the computational burden of our framework was exceptionally low. The overhead was either so small that it was statistically indistinguishable from baseline noise (as in Q2), or it was of a practically insignificant magnitude even when statistically detectable (as in Q1), making our scheme highly suitable for real-world deployment.

Communication Overhead. The primary overhead of our system lies in the increased data payload required to facilitate the blind signature protocol, as shown in Figs.~\ref{Figure:Q1 - Q5 Request Size differentials Comparison(14B)} and \ref{Figure:Q1 - Q5 Response Size differentials Comparison(14B)}. For the request, the size increased owing to the inclusion of the user's public key and the blinded message digest. For the response, particularly in the subscription mode, the size increased because it must contain the blind signature issued by the LLM provider. This communication overhead was a direct and necessary trade-off for establishing a verifiable and privacy-preserving communication channel. Given the critical security guarantees it provides, we argue that this overhead is acceptable and well-justified for the target applications. When comparing the two proposed modes, the subscription mode consistently incurred a larger data overhead than the API mode, because the former involved more complex interactions for maintaining subscription state and handling signature issuance.

\begin{figure}[!t]
	\centering
	\includegraphics [width=1\linewidth ]{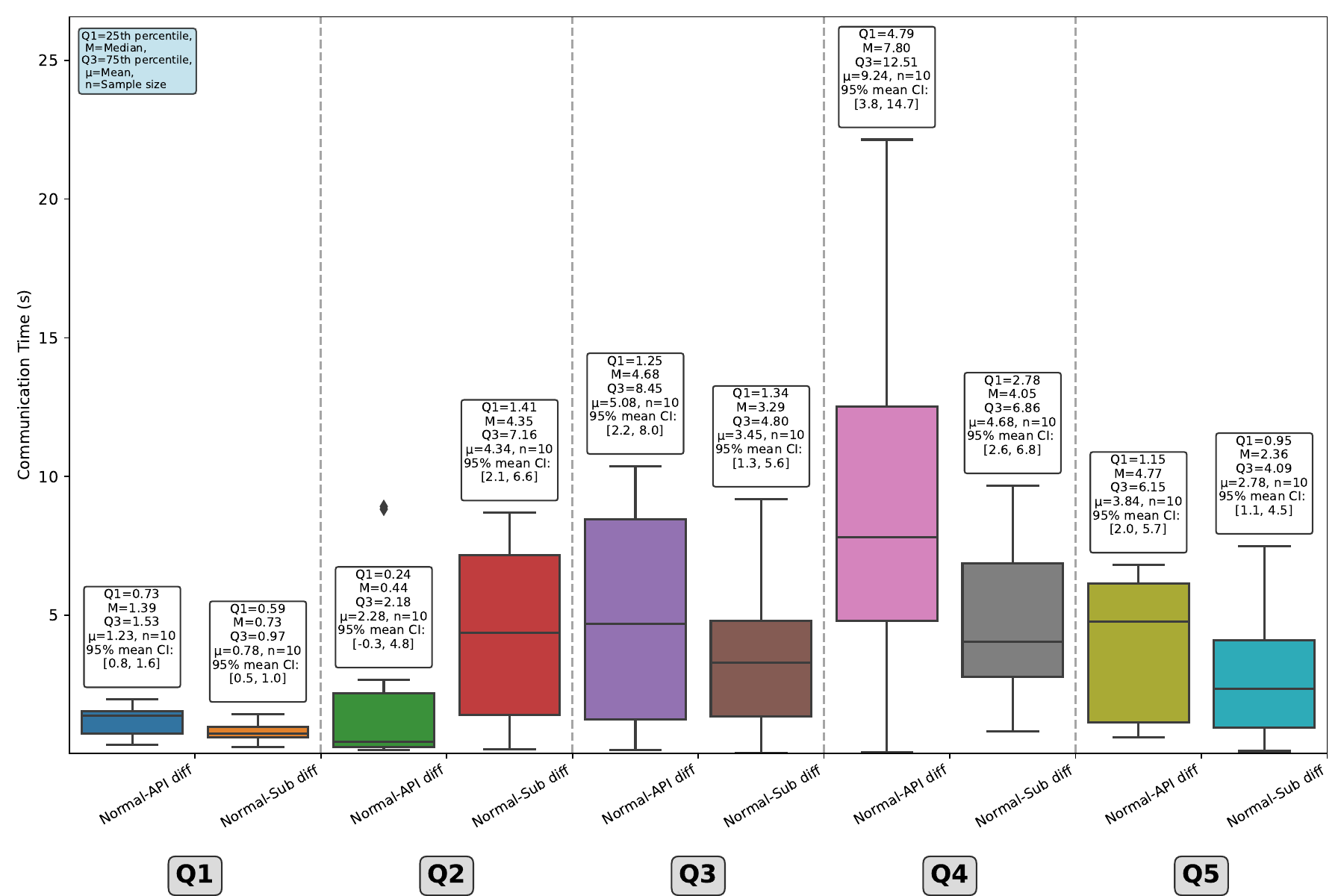}
	\caption{Q1 - Q5 communication time differentials comparison.}
	\label{Figure:Q1 - Q5 Communication Time differentials Comparison(14B)}
\end{figure}

\begin{figure}[!t]
	\centering
	\includegraphics [width=1\linewidth ]{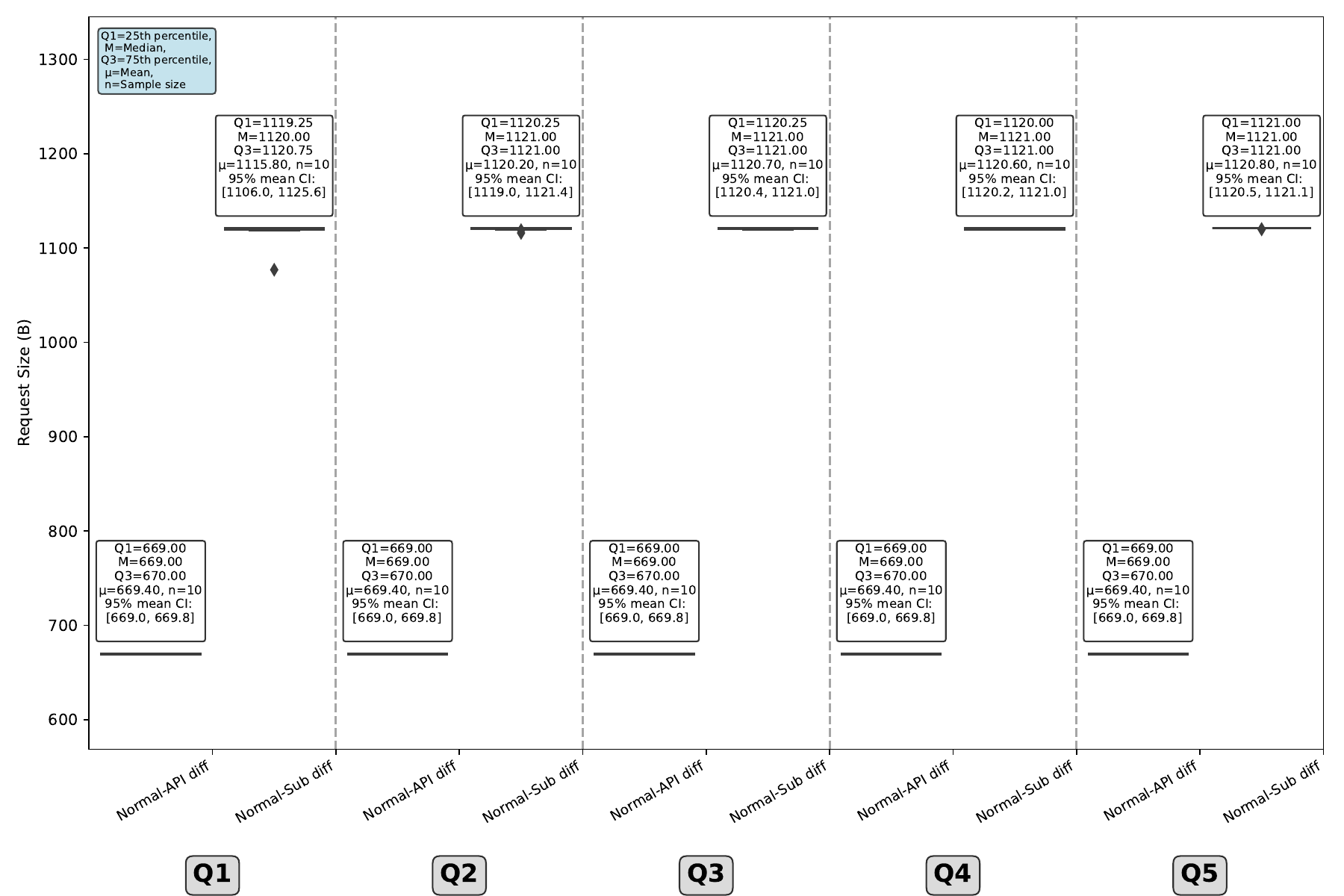}
	\caption{Q1 - Q5 communication request size differentials comparison.}
	\label{Figure:Q1 - Q5 Request Size differentials Comparison(14B)}
\end{figure}

\begin{figure}[!t]
	\centering
	\includegraphics [width=1\linewidth ]{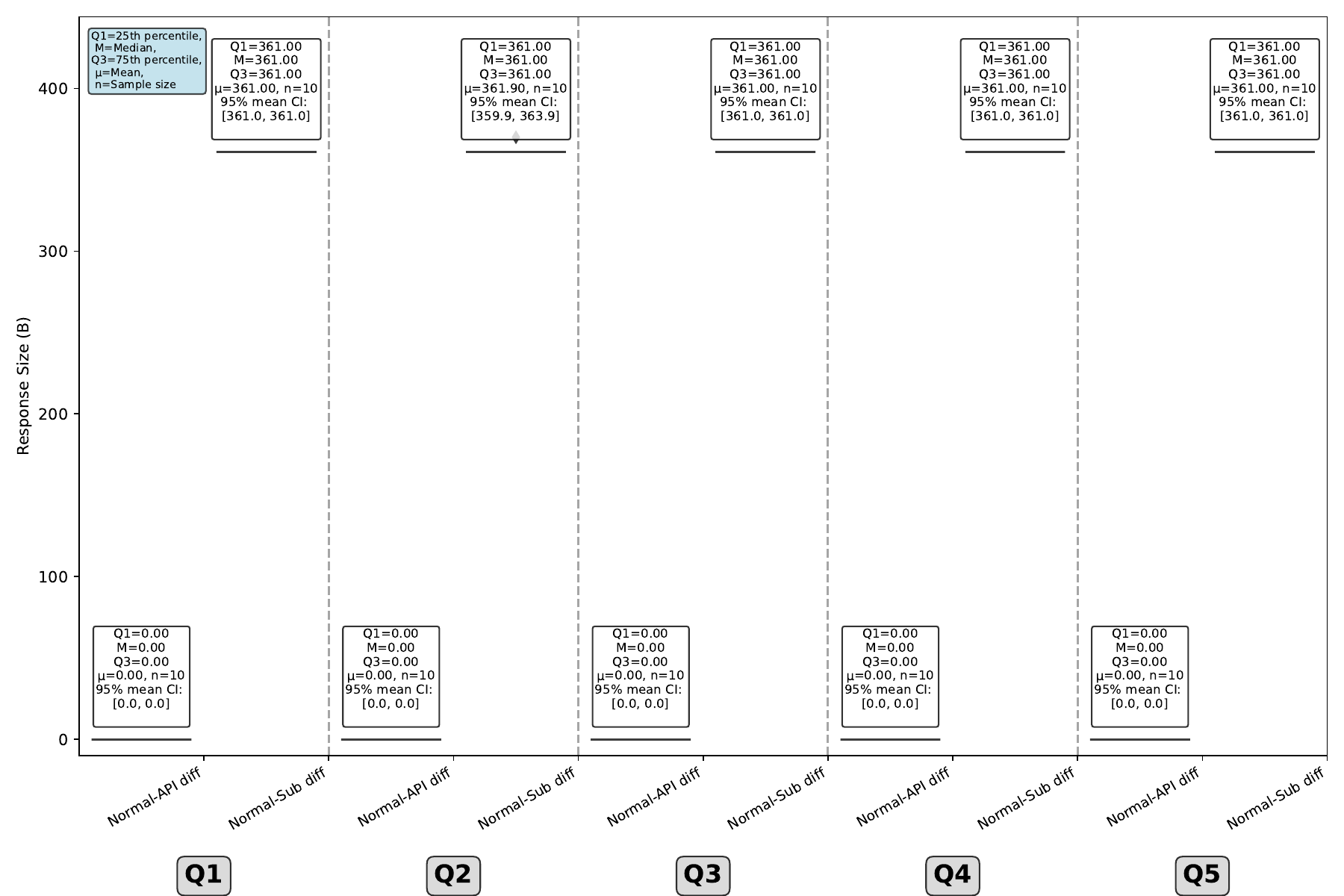}
	\caption{Q1 - Q5 communication response size differentials comparison.}
	\label{Figure:Q1 - Q5 Response Size differentials Comparison(14B)}
\end{figure}

In addition to the per-request overhead discussed above, we further evaluated whether the proposed framework exhibits stable behaviour under realistic multi-user concurrency—an essential characteristic for practical AIaaS deployments. Our experiments with different numbers of concurrent users show that the additional overhead introduced by our framework does not grow superlinearly with increasing user load. Although higher concurrency inevitably leads to slightly longer processing time, this increase is attributable to the underlying system’s need to maintain a task queue, which is a limitation of the deployment environment and hardware rather than of our framework.
A noteworthy observation is that, theoretically, the normal mode should incur latency equal to or lower than that of the API mode and strictly lower than that of the subscription mode. However, some of our measurements deviate from this expectation. This discrepancy arises from fluctuations in HTTPS transmission speed, as our framework relies on HTTPS-based communication. This discrepancy is primarily attributable to fluctuations in network performance. To verify this, we conducted TCP, UDP, and ping connectivity tests on servers within the same region, with the results summarised in Table \ref{tab:net_test_results}. Furthermore, because our framework communicates with the LLM provider over HTTPS at the application layer, we also evaluated temporal variations in HTTPS performance by downloading 10 MB, 100 MB, and 1000 MB files from the server at half-hour intervals. The corresponding measurements are presented in Fig.~\ref{Figure:Network Test}. The results show substantial temporal variability when downloading the same file over HTTPS, confirming that such inconsistencies stem from network characteristics rather than from our framework.

\begin{table*}[ht]
\centering
\caption{Network Test Result.}
\label{tab:net_test_results}
\resizebox{\textwidth}{!}{%
\begin{threeparttable}
    \begin{tabular}{c|c|c|c|c|c}
        \hline
         \textbf{Metric}&\textbf{Method} & \textbf{12:00-13:00} & \textbf{13:00-14:00} & \textbf{20:43-21:43} & \textbf{23:28-00:28} \\
        
        \hline
        \multirow{3}{*}{\textbf{Bandwidth}}
        &TCP UPLOAD(Mbits/s) &49.3 94.1 94.2  & 93.2  88.9 93.9& 90.2 94.3 94.1& 94.0 94.3 94.2\\
        
        \cline{2-6}
        &TCP DOWNLOAD(Mbits/s) &23.4  5.88 25.4& 6.2 70.0 52.6& 57.4 50.6 65.3& 56.2 57.5 73.6\\
        
        \cline{2-6}
        &UDP(Kbits/s)      &995 995 955& 995 995 995& 995 996 996& 996 996 995  \\
        
        \hline
        \multirow{3}{*}{\textbf{RTT (Round-Trip Time)(ms)}}& \multirow{3}{*}{Ping} & 
        Max=82&Max=82 &Max=81 &Max=80  \\
        
        \cline{3-6}
        & &
        Min=22  &Min=22 &Min=22 &Min=22 \\
        \cline{3-6}
        & & 
        Avg=24  & Avg=24 &Avg=23 &Avg=23 \\
        
        \hline
        \textbf{Packet Loss Rate(\textperthousand)}& Ping &23.6 & 6.1 & 0 & 1.4  \\
        
        \hline
    \end{tabular}
    \begin{tablenotes}  
      \small
      \item[*] It should be noted that each cell in the "Bandwidth" row contains three data points. This is because during our tests, we conducted three tests per hour, thus recording three sets of data. For instance, in the "TCP UPLOAD" row and the "12:00 - 13:00" column, 49.3 represents the test data from 12:00 to 12:20, 94.1 represents the data from 12:20 to 12:40, and 94.2 represents the data from 12:40 to 13:00. The rest of the data follows a similar pattern.
    \end{tablenotes}
\end{threeparttable}
}
\end{table*}
\begin{figure}[!htbp]
	\centering
	\includegraphics [width=0.8\linewidth ]{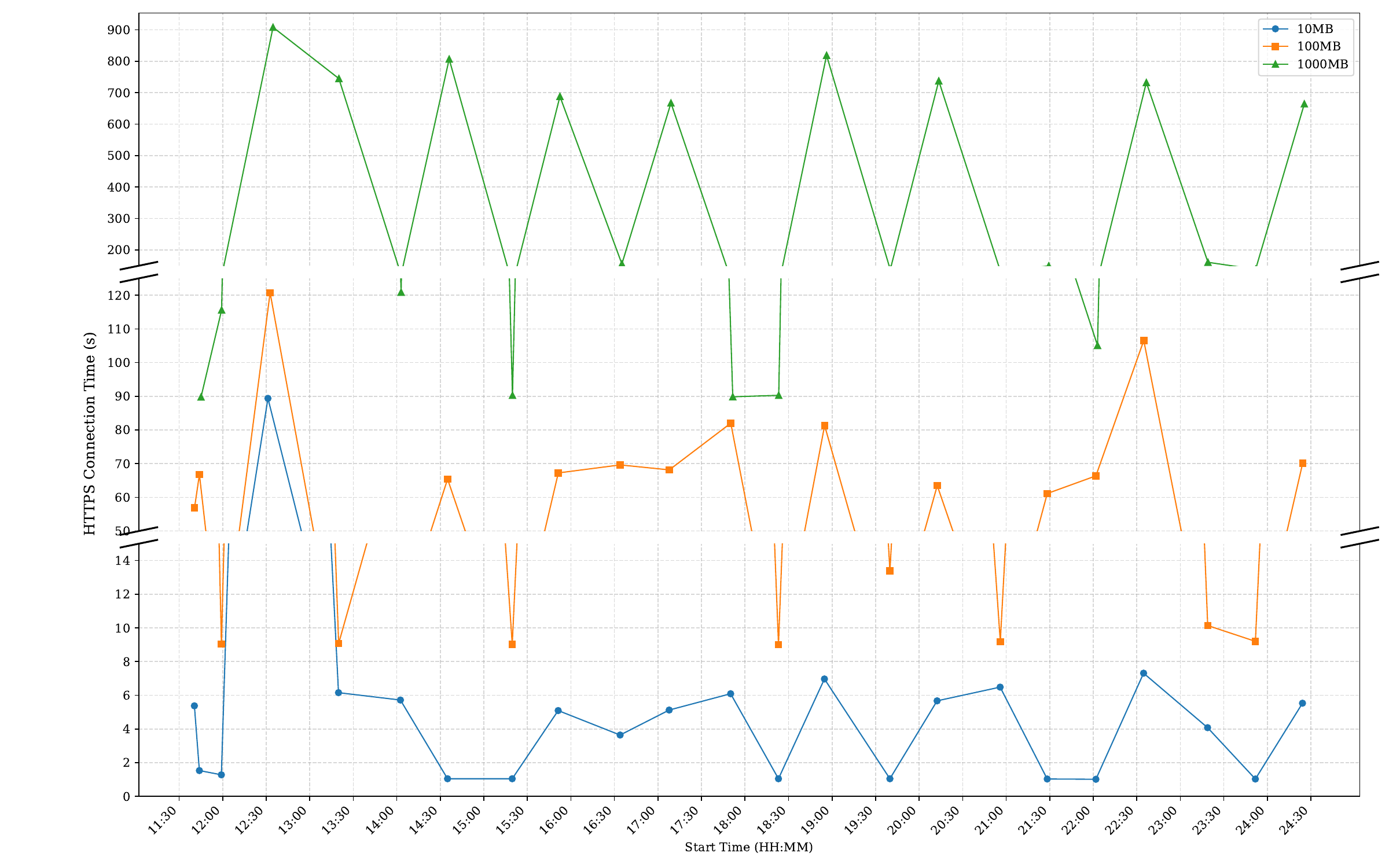}
	\caption{Network conditions for downloading 10MB ,100MB and 1000MB files via HTTPS protocol}
	\label{Figure:Network Test}
\end{figure}
Moreover, our evaluation across models from different institutions indicates that the choice of LLM affects only the absolute request time, request size, and response size, but not the relative differences among the normal, API, and subscription modes. These inter-mode differences originate solely from the signature operations introduced in our framework, demonstrating that the framework is model-agnostic and can be seamlessly applied to heterogeneous LLM providers.

Our previous experiment showed that the processing time and communication overhead required to generate multiple signatures were acceptable for practical use. Moreover, our framework allowed for flexible replacement of the underlying partially blind signature scheme. For instance, the RSA-based partially blind signature used in this study could be replaced with a batch partially blind signature scheme, which generated multiple signatures at once. Note that the extra communication overhead was not a performance bottleneck. Experimental results indicate that our frameworks introduced minimal differences in processing time. This demonstrates that our solution maintained user privacy without incurring significant resource consumption, achieving performance comparable to the original LLM.If readers analyse the experimental data of the DeepSeek-R1-Distill-Llama-8B model in the \ref{app1}, they will also reach the same conclusion. This further indicates that our solution is not dependent on the selection of the model.

Theoretically, compared to the original LLM, the additional overhead of our solution comes from processing the partially blind signatures between the client and server. During the blinding phase, one modular multiplication and one modular exponentiation operation are required. In the signing phase, two modular exponentiation operations are required. In the de-blinding phase, one modular multiplication and one modular exponentiation operation are required. The verification phase requires one modular exponentiation operation. Overall, a single signature incurs an additional cost of two modular multiplications and five modular exponentiations. In addition, hash operations are also required during the signature process; the hash function $H_M$ needs to be executed twice for each signature, while the $H_{MD}$ hash function needs to be executed only once within a certain period. Details of the partially blind signature scheme used in this study can be found in the study \cite{RSABlindSignatureswithPublicMetadata}.

\section{Conclusion}
This study presents a privacy-preserving framework for current LLM-based online service models, designed to ensure user anonymity. The framework supports both subscription-based and API-based service modes, using partially blind signatures to protect user identities while enabling service providers to manage access and prevent misuse, such as unauthorised usage. In addition, this framework functions as an additional layer that can be seamlessly integrated into existing systems. Our experimental results demonstrate that the proposed framework introduces minimal computational and communication overhead to the original system, maintaining efficiency with improved privacy.Meanwhile, since this framework is lightweight, it can be combined with other encryption schemes to achieve different functions.

\section*{Acknowledgment}

The authors would like to thank the editor and anonymous reviewers for
their constructive comments and suggestions.

\newpage

\printbibliography

@misc{rfc5280,
    series =    {Request for Comments},
    number =    5280,
    howpublished =  {RFC 5280},
    publisher = {RFC Editor},
    doi =       {10.17487/RFC5280},
    url =       {https://www.rfc-editor.org/info/rfc5280},
    author =    {Sharon Boeyen and Stefan Santesson and Tim Polk and Russ Housley and Stephen Farrell and David Cooper},
    title =     {{Internet X.509 Public Key Infrastructure Certificate and Certificate Revocation List (CRL) Profile}},
    pagetotal = 151,
    year =      2008,
    month =     may,
    options = {maxnames=3},
}

@article{806987,
  author       = {Radia J. Perlman},
  title        = {An overview of {PKI} trust models},
  journal      = {{IEEE} Netw.},
  volume       = {13},
  number       = {6},
  pages        = {38--43},
  year         = {1999},
  url          = {https://doi.org/10.1109/65.806987},
  doi          = {10.1109/65.806987},
  bibsource    = {dblp computer science bibliography, https://dblp.org}
}

@article{weise2001public,
  title={Public key infrastructure overview},
  author={Weise, Joel},
  journal={Sun BluePrints OnLine, August},
  pages={1--27},
  year={2001}
}

@article{zhao2023surveylargelanguagemodels,
  author       = {Wayne Xin Zhao and
                  Kun Zhou and
                  Junyi Li and
                  Tianyi Tang and
                  Xiaolei Wang and
                  Yupeng Hou and
                  Yingqian Min and
                  Beichen Zhang and
                  Junjie Zhang and
                  Zican Dong and
                  Yifan Du and
                  Chen Yang and
                  Yushuo Chen and
                  Zhipeng Chen and
                  Jinhao Jiang and
                  Ruiyang Ren and
                  Yifan Li and
                  Xinyu Tang and
                  Zikang Liu and
                  Peiyu Liu and
                  Jian{-}Yun Nie and
                  Ji{-}Rong Wen},
  title        = {A Survey of Large Language Models},
  journal      = {CoRR},
  volume       = {abs/2303.18223},
  year         = {2023},
  url          = {https://doi.org/10.48550/arXiv.2303.18223},
  doi          = {10.48550/ARXIV.2303.18223},
  eprinttype    = {arXiv},
  eprint       = {2303.18223},
  bibsource    = {dblp computer science bibliography, https://dblp.org}
}

@article{wei2022emergentabilitieslargelanguage,
  author       = {Jason Wei and
                  Yi Tay and
                  Rishi Bommasani and
                  Colin Raffel and
                  Barret Zoph and
                  Sebastian Borgeaud and
                  Dani Yogatama and
                  Maarten Bosma and
                  Denny Zhou and
                  Donald Metzler and
                  Ed H. Chi and
                  Tatsunori Hashimoto and
                  Oriol Vinyals and
                  Percy Liang and
                  Jeff Dean and
                  William Fedus},
  title        = {Emergent Abilities of Large Language Models},
  journal      = {Trans. Mach. Learn. Res.},
  volume       = {2022},
  year         = {2022},
  url          = {https://openreview.net/forum?id=yzkSU5zdwD},
  bibsource    = {dblp computer science bibliography, https://dblp.org}
}

@article{openai2024gpt4technicalreport,
  author       = {OpenAI},
  title        = {{GPT-4} Technical Report},
  journal      = {CoRR},
  volume       = {abs/2303.08774},
  year         = {2023},
  url          = {https://doi.org/10.48550/arXiv.2303.08774},
  doi          = {10.48550/ARXIV.2303.08774},
  eprinttype    = {arXiv},
  eprint       = {2303.08774},
  bibsource    = {dblp computer science bibliography, https://dblp.org}
}

@article{hadi2023survey,
  title={A survey on large language models: Applications, challenges, limitations, and practical usage},
  author={Hadi, Muhammad Usman and Qureshi, Rizwan and Shah, Abbas and Irfan, Muhammad and Zafar, Anas and Shaikh, Muhammad Bilal and Akhtar, Naveed and Wu, Jia and Mirjalili, Seyedali and others},
  journal={Authorea Preprints},
  year={2023},
  publisher={Authorea}
}

@article{lo2023clear,
title = {The CLEAR path: A framework for enhancing information literacy through prompt engineering},
journal = {The Journal of Academic Librarianship},
volume = {49},
number = {4},
pages = {102720},
year = {2023},
issn = {0099-1333},
doi = {https://doi.org/10.1016/j.acalib.2023.102720},
url = {https://www.sciencedirect.com/science/article/pii/S0099133323000599},
author = {Leo S. Lo}
}

@article{wang2023prompt,
  author       = {Jiaqi Wang and
                  Enze Shi and
                  Sigang Yu and
                  Zihao Wu and
                  Chong Ma and
                  Haixing Dai and
                  Qiushi Yang and
                  Yanqing Kang and
                  Jinru Wu and
                  Huawen Hu and
                  Chenxi Yue and
                  Haiyang Zhang and
                  Yiheng Liu and
                  Xiang Li and
                  Bao Ge and
                  Dajiang Zhu and
                  Yixuan Yuan and
                  Dinggang Shen and
                  Tianming Liu and
                  Shu Zhang},
  title        = {Prompt Engineering for Healthcare: Methodologies and Applications},
  journal      = {CoRR},
  volume       = {abs/2304.14670},
  year         = {2023},
  url          = {https://doi.org/10.48550/arXiv.2304.14670},
  doi          = {10.48550/ARXIV.2304.14670},
  eprinttype    = {arXiv},
  eprint       = {2304.14670},
  bibsource    = {dblp computer science bibliography, https://dblp.org}
}

@inproceedings{chaum1983blind,
  author       = {David Chaum},
  editor       = {David Chaum and
                  Ronald L. Rivest and
                  Alan T. Sherman},
  title        = {Blind Signatures for Untraceable Payments},
  booktitle    = {Advances in Cryptology: Proceedings of {CRYPTO} '82, Santa Barbara,
                  California, USA, August 23-25, 1982},
  pages        = {199--203},
  publisher    = {Plenum Press, New York},
  year         = {1982},
  url          = {https://doi.org/10.1007/978-1-4757-0602-4\_18},
  doi          = {10.1007/978-1-4757-0602-4\_18},
  bibsource    = {dblp computer science bibliography, https://dblp.org}
}

@article{khater2018blind,
author = {Khater, Monira M. and Al-Ahwal, Ayman and Selim, Mazen M. and Zayed, Hala H.},
title = { Blind Signature Schemes based on ElGamal Signature for Electronic Voting: A Survey },
journal = { International Journal of Computer Applications },
issue_date = { Apr 2018 },
volume = { 180 },
number = { 30 },
month = { Apr },
year = { 2018 },
issn = { 0975-8887 },
pages = { 21-28 },
numpages = {9},
url = { https://ijcaonline.org/archives/volume180/number30/29234-2018916766/ },
doi = { 10.5120/ijca2018916766 },
publisher = {Foundation of Computer Science (FCS), NY, USA},
address = {New York, USA}
}

@inproceedings{abe1996date,
  author       = {Masayuki Abe and
                  Eiichiro Fujisaki},
  editor       = {Kwangjo Kim and
                  Tsutomu Matsumoto},
  title        = {How to Date Blind Signatures},
  booktitle    = {Advances in Cryptology - {ASIACRYPT} '96, International Conference
                  on the Theory and Applications of Cryptology and Information Security,
                  Kyongju, Korea, November 3-7, 1996, Proceedings},
  series       = {Lecture Notes in Computer Science},
  volume       = {1163},
  pages        = {244--251},
  publisher    = {Springer},
  year         = {1996},
  url          = {https://doi.org/10.1007/BFb0034851},
  doi          = {10.1007/BFB0034851},
  bibsource    = {dblp computer science bibliography, https://dblp.org}
}

@inproceedings{chaum1990untraceable,
  author       = {David Chaum and
                  Amos Fiat and
                  Moni Naor},
  editor       = {Shafi Goldwasser},
  title        = {Untraceable Electronic Cash},
  booktitle    = {Advances in Cryptology - {CRYPTO} '88, 8th Annual International Cryptology
                  Conference, Santa Barbara, California, USA, August 21-25, 1988, Proceedings},
  series       = {Lecture Notes in Computer Science},
  volume       = {403},
  pages        = {319--327},
  publisher    = {Springer},
  year         = {1988},
  url          = {https://doi.org/10.1007/0-387-34799-2\_25},
  doi          = {10.1007/0-387-34799-2\_25},
  bibsource    = {dblp computer science bibliography, https://dblp.org}
}

@misc{rfc8446,
    series =    {Request for Comments},
    number =    8446,
    howpublished =  {RFC 8446},
    publisher = {RFC Editor},
    doi =       {10.17487/RFC8446},
    url =       {https://www.rfc-editor.org/info/rfc8446},
    author =    {Eric Rescorla},
    title =     {{The Transport Layer Security (TLS) Protocol Version 1.3}},
    pagetotal = 160,
    year =      2018,
    month =     aug,
}

@misc{rfc2818,
    series =    {Request for Comments},
    number =    2818,
    howpublished =  {RFC 2818},
    publisher = {RFC Editor},
    doi =       {10.17487/RFC2818},
    url =       {https://www.rfc-editor.org/info/rfc2818},
    author =    {Eric Rescorla},
    title =     {{HTTP Over TLS}},
    pagetotal = 7,
    year =      2000,
    month =     may,
}

@article{singhal2022largelanguagemodelsencode,
  author       = {Karan Singhal and
                  Shekoofeh Azizi and
                  Tao Tu and
                  S. Sara Mahdavi and
                  Jason Wei and
                  Hyung Won Chung and
                  Nathan Scales and
                  Ajay Kumar Tanwani and
                  Heather Cole{-}Lewis and
                  Stephen Pfohl and
                  Perry Payne and
                  Martin Seneviratne and
                  Paul Gamble and
                  Chris Kelly and
                  Nathaneal Sch{\"{a}}rli and
                  Aakanksha Chowdhery and
                  Philip Andrew Mansfield and
                  Blaise Ag{\"{u}}era y Arcas and
                  Dale R. Webster and
                  Gregory S. Corrado and
                  Yossi Matias and
                  Katherine Chou and
                  Juraj Gottweis and
                  Nenad Tomasev and
                  Yun Liu and
                  Alvin Rajkomar and
                  Joelle K. Barral and
                  Christopher Semturs and
                  Alan Karthikesalingam and
                  Vivek Natarajan},
  title        = {Large Language Models Encode Clinical Knowledge},
  journal      = {CoRR},
  volume       = {abs/2212.13138},
  year         = {2022},
  url          = {https://doi.org/10.48550/arXiv.2212.13138},
  doi          = {10.48550/ARXIV.2212.13138},
  eprinttype    = {arXiv},
  eprint       = {2212.13138},
  bibsource    = {dblp computer science bibliography, https://dblp.org}
}

@inproceedings{zhang2024simulatingclassroomeducationllmempowered,
  author       = {Zheyuan Zhang and
                  Daniel Zhang{-}Li and
                  Jifan Yu and
                  Linlu Gong and
                  Jinchang Zhou and
                  Zhanxin Hao and
                  Jianxiao Jiang and
                  Jie Cao and
                  Huiqin Liu and
                  Zhiyuan Liu and
                  Lei Hou and
                  Juanzi Li},
  editor       = {Luis Chiruzzo and
                  Alan Ritter and
                  Lu Wang},
  title        = {Simulating Classroom Education with LLM-Empowered Agents},
  booktitle    = {Proceedings of the 2025 Conference of the Nations of the Americas
                  Chapter of the Association for Computational Linguistics: Human Language
                  Technologies, {NAACL} 2025 - Volume 1: Long Papers, Albuquerque, New
                  Mexico, USA, April 29 - May 4, 2025},
  pages        = {10364--10379},
  publisher    = {Association for Computational Linguistics},
  year         = {2025},
  url          = {https://doi.org/10.18653/v1/2025.naacl-long.520},
  doi          = {10.18653/V1/2025.NAACL-LONG.520},
  bibsource    = {dblp computer science bibliography, https://dblp.org}
}

@inproceedings{yu2023scalingrobotlearningsemantically,
  author       = {Tianhe Yu and
                  Ted Xiao and
                  Jonathan Tompson and
                  Austin Stone and
                  Su Wang and
                  Anthony Brohan and
                  Jaspiar Singh and
                  Clayton Tan and
                  Dee M and
                  Jodilyn Peralta and
                  Karol Hausman and
                  Brian Ichter and
                  Fei Xia},
  editor       = {Kostas E. Bekris and
                  Kris Hauser and
                  Sylvia L. Herbert and
                  Jingjin Yu},
  title        = {Scaling Robot Learning with Semantically Imagined Experience},
  booktitle    = {Robotics: Science and Systems XIX, Daegu, Republic of Korea, July
                  10-14, 2023},
  year         = {2023},
  url          = {https://doi.org/10.15607/RSS.2023.XIX.027},
  doi          = {10.15607/RSS.2023.XIX.027},
  bibsource    = {dblp computer science bibliography, https://dblp.org}
}

@article{hu2023lookleapunveilingpower,
  author       = {Yingdong Hu and
                  Fanqi Lin and
                  Tong Zhang and
                  Li Yi and
                  Yang Gao},
  title        = {Look Before You Leap: Unveiling the Power of {GPT-4V} in Robotic Vision-Language
                  Planning},
  journal      = {CoRR},
  volume       = {abs/2311.17842},
  year         = {2023},
  url          = {https://doi.org/10.48550/arXiv.2311.17842},
  doi          = {10.48550/ARXIV.2311.17842},
  eprinttype    = {arXiv},
  eprint       = {2311.17842},
  bibsource    = {dblp computer science bibliography, https://dblp.org}
}

@article{dubey2024llama3herdmodels,
  author       = {Abhimanyu Dubey and
                  Abhinav Jauhri and
                  Abhinav Pandey and
                  Abhishek Kadian and
                  Ahmad Al{-}Dahle and
                  Aiesha Letman and
                  Akhil Mathur and
                  Alan Schelten and
                  Amy Yang and
                  Angela Fan and
                  Anirudh Goyal and
                  Anthony Hartshorn and
                  Aobo Yang and
                  Archi Mitra and
                  Archie Sravankumar and
                  Artem Korenev and
                  Arthur Hinsvark and
                  Arun Rao and
                  Aston Zhang and
                  Aur{\'{e}}lien Rodriguez and
                  Austen Gregerson and
                  Ava Spataru and
                  Baptiste Rozi{\`{e}}re and
                  Bethany Biron and
                  Binh Tang and
                  Bobbie Chern and
                  Charlotte Caucheteux and
                  Chaya Nayak and
                  Chloe Bi and
                  Chris Marra and
                  Chris McConnell and
                  Christian Keller and
                  Christophe Touret and
                  Chunyang Wu and
                  Corinne Wong and
                  Cristian Canton Ferrer and
                  Cyrus Nikolaidis and
                  Damien Allonsius and
                  Daniel Song and
                  Danielle Pintz and
                  Danny Livshits and
                  David Esiobu and
                  Dhruv Choudhary and
                  Dhruv Mahajan and
                  Diego Garcia{-}Olano and
                  Diego Perino and
                  Dieuwke Hupkes and
                  Egor Lakomkin and
                  Ehab AlBadawy and
                  Elina Lobanova and
                  Emily Dinan and
                  Eric Michael Smith and
                  Filip Radenovic and
                  Frank Zhang and
                  Gabriel Synnaeve and
                  Gabrielle Lee and
                  Georgia Lewis Anderson and
                  Graeme Nail and
                  Gr{\'{e}}goire Mialon and
                  Guan Pang and
                  Guillem Cucurell and
                  Hailey Nguyen and
                  Hannah Korevaar and
                  Hu Xu and
                  Hugo Touvron and
                  Iliyan Zarov and
                  Imanol Arrieta Ibarra and
                  Isabel M. Kloumann and
                  Ishan Misra and
                  Ivan Evtimov and
                  Jade Copet and
                  Jaewon Lee and
                  Jan Geffert and
                  Jana Vranes and
                  Jason Park and
                  Jay Mahadeokar and
                  Jeet Shah and
                  Jelmer van der Linde and
                  Jennifer Billock and
                  Jenny Hong and
                  Jenya Lee and
                  Jeremy Fu and
                  Jianfeng Chi and
                  Jianyu Huang and
                  Jiawen Liu and
                  Jie Wang and
                  Jiecao Yu and
                  Joanna Bitton and
                  Joe Spisak and
                  Jongsoo Park and
                  Joseph Rocca and
                  Joshua Johnstun and
                  Joshua Saxe and
                  Junteng Jia and
                  Kalyan Vasuden Alwala and
                  Kartikeya Upasani and
                  Kate Plawiak and
                  Ke Li and
                  Kenneth Heafield and
                  Kevin Stone and
                  et al.},
  title        = {The Llama 3 Herd of Models},
  journal      = {CoRR},
  volume       = {abs/2407.21783},
  year         = {2024},
  url          = {https://doi.org/10.48550/arXiv.2407.21783},
  doi          = {10.48550/ARXIV.2407.21783},
  eprinttype    = {arXiv},
  eprint       = {2407.21783},
  bibsource    = {dblp computer science bibliography, https://dblp.org}
}

@inproceedings{hao2022iron,
  author       = {Meng Hao and
                  Hongwei Li and
                  Hanxiao Chen and
                  Pengzhi Xing and
                  Guowen Xu and
                  Tianwei Zhang},
  editor       = {Sanmi Koyejo and
                  S. Mohamed and
                  A. Agarwal and
                  Danielle Belgrave and
                  K. Cho and
                  A. Oh},
  title        = {Iron: Private Inference on Transformers},
  booktitle    = {Advances in Neural Information Processing Systems 35: Annual Conference
                  on Neural Information Processing Systems 2022, NeurIPS 2022, New Orleans,
                  LA, USA, November 28 - December 9, 2022},
  year         = {2022},
  url          = {http://papers.nips.cc/paper\_files/paper/2022/hash/64e2449d74f84e5b1a5c96ba7b3d308e-Abstract-Conference.html},
  bibsource    = {dblp computer science bibliography, https://dblp.org}
}

@inproceedings{cryptoeprint:2023/1678,
  author       = {Wen{-}jie Lu and
                  Zhicong Huang and
                  Zhen Gu and
                  Jingyu Li and
                  Jian Liu and
                  Cheng Hong and
                  Kui Ren and
                  Tao Wei and
                  Wenguang Chen},
  title        = {BumbleBee: Secure Two-party Inference Framework for Large Transformers},
  booktitle    = {32nd Annual Network and Distributed System Security Symposium, {NDSS}
                  2025, San Diego, California, USA, February 24-28, 2025},
  publisher    = {The Internet Society},
  year         = {2025},
  url          = {https://www.ndss-symposium.org/ndss-paper/bumblebee-secure-two-party-inference-framework-for-large-transformers/},
  bibsource    = {dblp computer science bibliography, https://dblp.org}
}

@article{ding2023eastefficientaccuratesecure,
  author       = {Yuanchao Ding and
                  Hua Guo and
                  Yewei Guan and
                  Weixin Liu and
                  Jiarong Huo and
                  Zhenyu Guan and
                  Xiyong Zhang},
  title        = {East: Efficient and Accurate Secure Inference Framework for Transformer},
  journal      = {{IEEE} Trans. Serv. Comput.},
  volume       = {18},
  number       = {4},
  pages        = {2038--2046},
  year         = {2025},
  url          = {https://doi.org/10.1109/TSC.2025.3577491},
  doi          = {10.1109/TSC.2025.3577491},
  bibsource    = {dblp computer science bibliography, https://dblp.org}
}

@inproceedings{10190506,
  author       = {Yoshimasa Akimoto and
                  Kazuto Fukuchi and
                  Youhei Akimoto and
                  Jun Sakuma},
  title        = {Privformer: Privacy-preserving Transformer with {MPC}},
  booktitle    = {8th {IEEE} European Symposium on Security and Privacy, EuroS{\&}P
                  2023, Delft, Netherlands, July 3-7, 2023},
  pages        = {392--410},
  publisher    = {{IEEE}},
  year         = {2023},
  url          = {https://doi.org/10.1109/EuroSP57164.2023.00031},
  doi          = {10.1109/EUROSP57164.2023.00031},
  bibsource    = {dblp computer science bibliography, https://dblp.org}
}

@article{gupta2023sigma,
  author       = {Kanav Gupta and
                  Neha Jawalkar and
                  Ananta Mukherjee and
                  Nishanth Chandran and
                  Divya Gupta and
                  Ashish Panwar and
                  Rahul Sharma},
  title        = {{SIGMA:} Secure {GPT} Inference with Function Secret Sharing},
  journal      = {Proc. Priv. Enhancing Technol.},
  volume       = {2024},
  number       = {4},
  pages        = {61--79},
  year         = {2024},
  url          = {https://doi.org/10.56553/popets-2024-0107},
  doi          = {10.56553/POPETS-2024-0107},
  bibsource    = {dblp computer science bibliography, https://dblp.org}
}

@inproceedings{brown2020languagemodelsfewshotlearners,
  author       = {Tom B. Brown and
                  Benjamin Mann and
                  Nick Ryder and
                  Melanie Subbiah and
                  Jared Kaplan and
                  Prafulla Dhariwal and
                  Arvind Neelakantan and
                  Pranav Shyam and
                  Girish Sastry and
                  Amanda Askell and
                  Sandhini Agarwal and
                  Ariel Herbert{-}Voss and
                  Gretchen Krueger and
                  Tom Henighan and
                  Rewon Child and
                  Aditya Ramesh and
                  Daniel M. Ziegler and
                  Jeffrey Wu and
                  Clemens Winter and
                  Christopher Hesse and
                  Mark Chen and
                  Eric Sigler and
                  Mateusz Litwin and
                  Scott Gray and
                  Benjamin Chess and
                  Jack Clark and
                  Christopher Berner and
                  Sam McCandlish and
                  Alec Radford and
                  Ilya Sutskever and
                  Dario Amodei},
  editor       = {Hugo Larochelle and
                  Marc'Aurelio Ranzato and
                  Raia Hadsell and
                  Maria{-}Florina Balcan and
                  Hsuan{-}Tien Lin},
  title        = {Language Models are Few-Shot Learners},
  booktitle    = {Advances in Neural Information Processing Systems 33: Annual Conference
                  on Neural Information Processing Systems 2020, NeurIPS 2020, December
                  6-12, 2020, virtual},
  year         = {2020},
  url          = {https://proceedings.neurips.cc/paper/2020/hash/1457c0d6bfcb4967418bfb8ac142f64a-Abstract.html},
  bibsource    = {dblp computer science bibliography, https://dblp.org}
}

@article{achiam2023gpt,
  author       = {OpenAI},
  title        = {{GPT-4} Technical Report},
  journal      = {CoRR},
  volume       = {abs/2303.08774},
  year         = {2023},
  url          = {https://doi.org/10.48550/arXiv.2303.08774},
  doi          = {10.48550/ARXIV.2303.08774},
  eprinttype    = {arXiv},
  eprint       = {2303.08774},
  bibsource    = {dblp computer science bibliography, https://dblp.org}
}

@inproceedings{chen2022thexprivacypreservingtransformerinference,
  author       = {Tianyu Chen and
                  Hangbo Bao and
                  Shaohan Huang and
                  Li Dong and
                  Binxing Jiao and
                  Daxin Jiang and
                  Haoyi Zhou and
                  Jianxin Li and
                  Furu Wei},
  editor       = {Smaranda Muresan and
                  Preslav Nakov and
                  Aline Villavicencio},
  title        = {{THE-X:} Privacy-Preserving Transformer Inference with Homomorphic
                  Encryption},
  booktitle    = {Findings of the Association for Computational Linguistics: {ACL} 2022,
                  Dublin, Ireland, May 22-27, 2022},
  pages        = {3510--3520},
  publisher    = {Association for Computational Linguistics},
  year         = {2022},
  url          = {https://doi.org/10.18653/v1/2022.findings-acl.277},
  doi          = {10.18653/V1/2022.FINDINGS-ACL.277},
  bibsource    = {dblp computer science bibliography, https://dblp.org}
}

@inproceedings{akimoto2023privformer,
  author       = {Yoshimasa Akimoto and
                  Kazuto Fukuchi and
                  Youhei Akimoto and
                  Jun Sakuma},
  title        = {Privformer: Privacy-preserving Transformer with {MPC}},
  booktitle    = {8th {IEEE} European Symposium on Security and Privacy, EuroS{\&}P
                  2023, Delft, Netherlands, July 3-7, 2023},
  pages        = {392--410},
  publisher    = {{IEEE}},
  year         = {2023},
  url          = {https://doi.org/10.1109/EuroSP57164.2023.00031},
  doi          = {10.1109/EUROSP57164.2023.00031},
  bibsource    = {dblp computer science bibliography, https://dblp.org}
}

@misc{majmudar2022differentiallyprivatedecodinglarge,
      title={Differentially Private Decoding in Large Language Models}, 
      author={Jimit Majmudar and Christophe Dupuy and Charith Peris and Sami Smaili and Rahul Gupta and Richard Zemel},
      year={2022},
      eprint={2205.13621},
      archivePrefix={arXiv},
      primaryClass={cs.CL},
      url={https://arxiv.org/abs/2205.13621}, 
}

@inproceedings{du2023dp,
  author       = {Minxin Du and
                  Xiang Yue and
                  Sherman S. M. Chow and
                  Tianhao Wang and
                  Chenyu Huang and
                  Huan Sun},
  editor       = {Weizhi Meng and
                  Christian Damsgaard Jensen and
                  Cas Cremers and
                  Engin Kirda},
  title        = {DP-Forward: Fine-tuning and Inference on Language Models with Differential
                  Privacy in Forward Pass},
  booktitle    = {Proceedings of the 2023 {ACM} {SIGSAC} Conference on Computer and
                  Communications Security, {CCS} 2023, Copenhagen, Denmark, November
                  26-30, 2023},
  pages        = {2665--2679},
  publisher    = {{ACM}},
  year         = {2023},
  url          = {https://doi.org/10.1145/3576915.3616592},
  doi          = {10.1145/3576915.3616592},
  bibsource    = {dblp computer science bibliography, https://dblp.org}
}

@misc{ladd2012blind,
  title={Blind signatures for bitcoin transaction anonymity},
  author={Ladd, Watsonn},
  year={2012},
  publisher={Tech. Rep., 2013 [Online]. Available: http://wbl. github. io/bitcoinanon. pdf}
}

@inproceedings{Asghar2015ASO,
  title={A Survey on Blind Digital Signatures},
  author={Nabiha Asghar},
  year={2015},
  url={https://api.semanticscholar.org/CorpusID:44319470}
}

@inproceedings{Securityofblinddigitalsignatures,
  author       = {Ari Juels and
                  Michael Luby and
                  Rafail Ostrovsky},
  editor       = {Burton S. Kaliski Jr.},
  title        = {Security of Blind Digital Signatures (Extended Abstract)},
  booktitle    = {Advances in Cryptology - {CRYPTO} '97, 17th Annual International Cryptology
                  Conference, Santa Barbara, California, USA, August 17-21, 1997, Proceedings},
  series       = {Lecture Notes in Computer Science},
  volume       = {1294},
  pages        = {150--164},
  publisher    = {Springer},
  year         = {1997},
  url          = {https://doi.org/10.1007/BFb0052233},
  doi          = {10.1007/BFB0052233},
  bibsource    = {dblp computer science bibliography, https://dblp.org}
}

@article{RSABlindSignatureswithPublicMetadata,
  author       = {Ghous Amjad and
                  Kevin Yeo and
                  Moti Yung},
  title        = {{RSA} Blind Signatures with Public Metadata},
  journal      = {Proc. Priv. Enhancing Technol.},
  volume       = {2025},
  number       = {1},
  pages        = {37--57},
  year         = {2025},
  url          = {https://doi.org/10.56553/popets-2025-0004},
  doi          = {10.56553/POPETS-2025-0004},
  bibsource    = {dblp computer science bibliography, https://dblp.org}
}

@techreport{amjad-cfrg-partially-blind-rsa-03,
    number =    {draft-amjad-cfrg-partially-blind-rsa-03},
    type =      {Internet-Draft},
    institution =   {Internet Engineering Task Force},
    publisher = {Internet Engineering Task Force},
    note =      {Work in Progress},
    url =       {https://datatracker.ietf.org/doc/draft-amjad-cfrg-partially-blind-rsa/03/},
    author =    {Ghous Ali Amjad and Scott Hendrickson and Christopher A. Wood and Kevin W. L. Yeo},
    title =     {{Partially Blind RSA Signatures}},
    pagetotal = 24,
    year =      2024,
    month =     aug,
    day =       15,
}

\newpage
 \appendix
 \section{DeepSeek-R1-Distill-Llama-8B model experimental result graph}
 \label{app1}
 
\begin{figure}[!htbp]
	\centering
	\includegraphics [width=0.8\linewidth ]{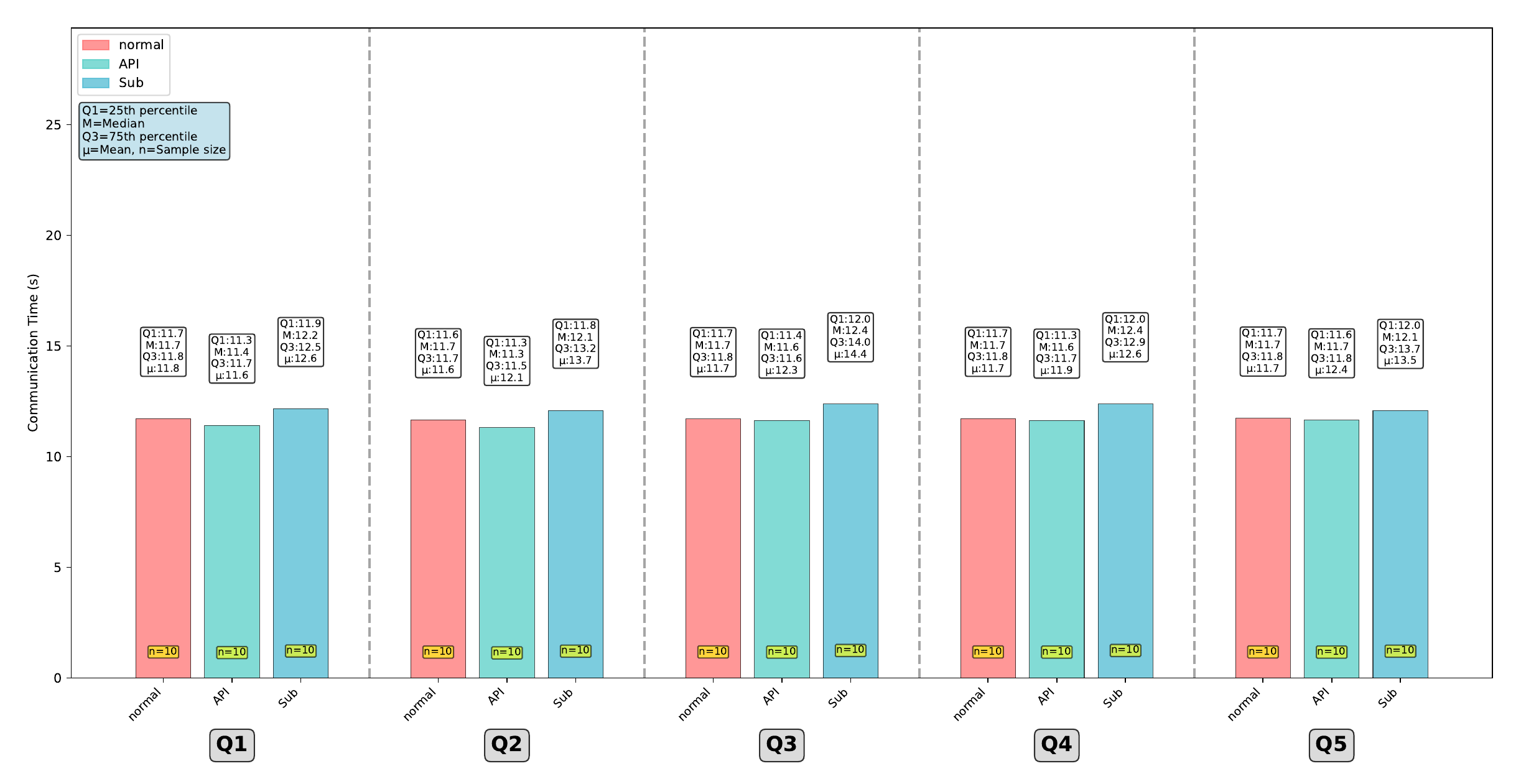}
	\caption{Q1 - Q5 communication time comparison(DeepSeek-R1-Distill-Llama-8B).}
	\label{Figure:Q1 - Q5 Communication Time Comparison(8B)}
\end{figure}

\begin{figure}[!htbp]
	\centering
	\includegraphics [width=0.8\linewidth ]{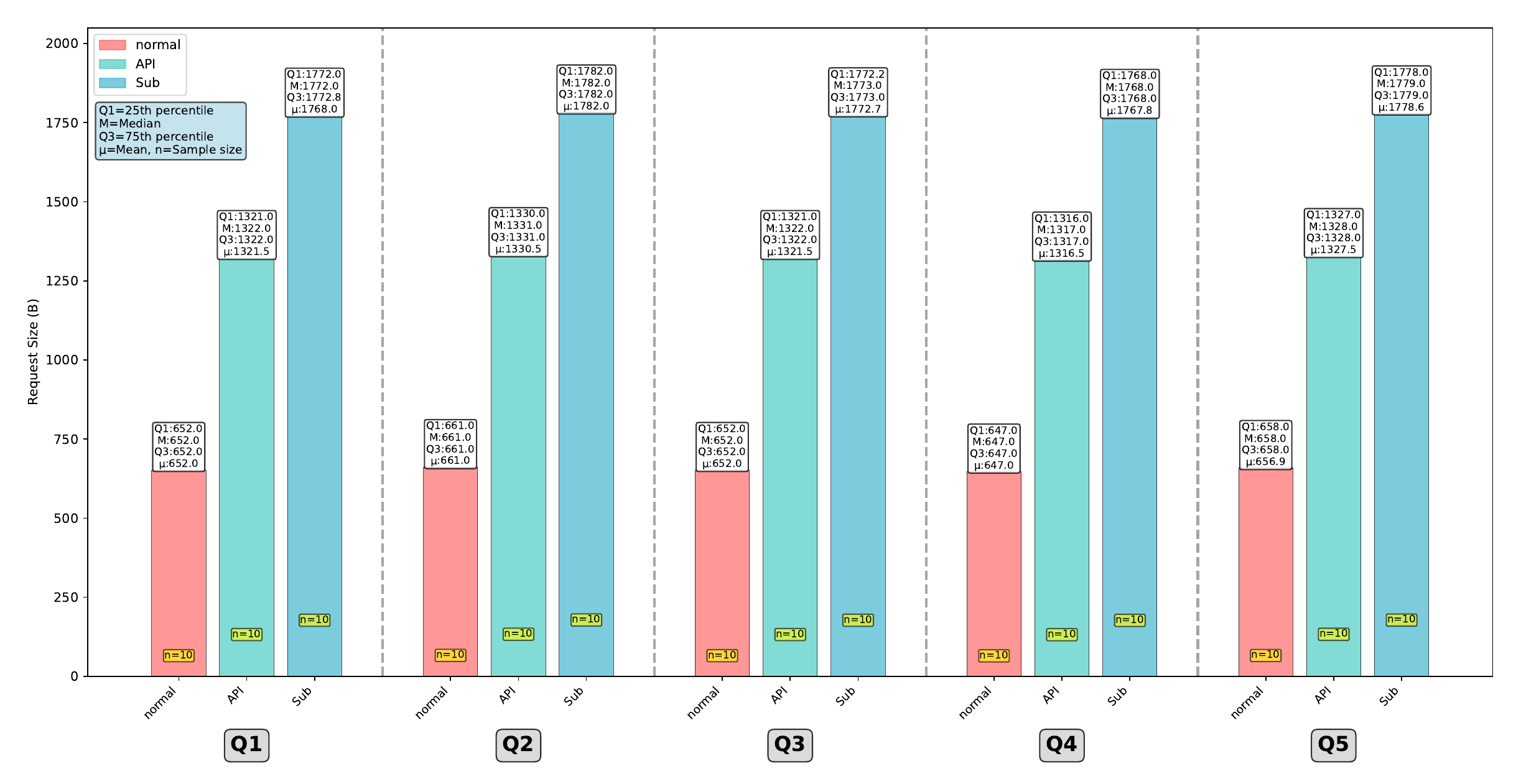}
	\caption{Q1 - Q5 communication  request size comparison(DeepSeek-R1-Distill-Llama-8B).}
	\label{Figure:Q1 - Q5 Request Size Comparison(8B)}
\end{figure}

\begin{figure}[!htbp]
	\centering
	\includegraphics [width=0.8\linewidth ]{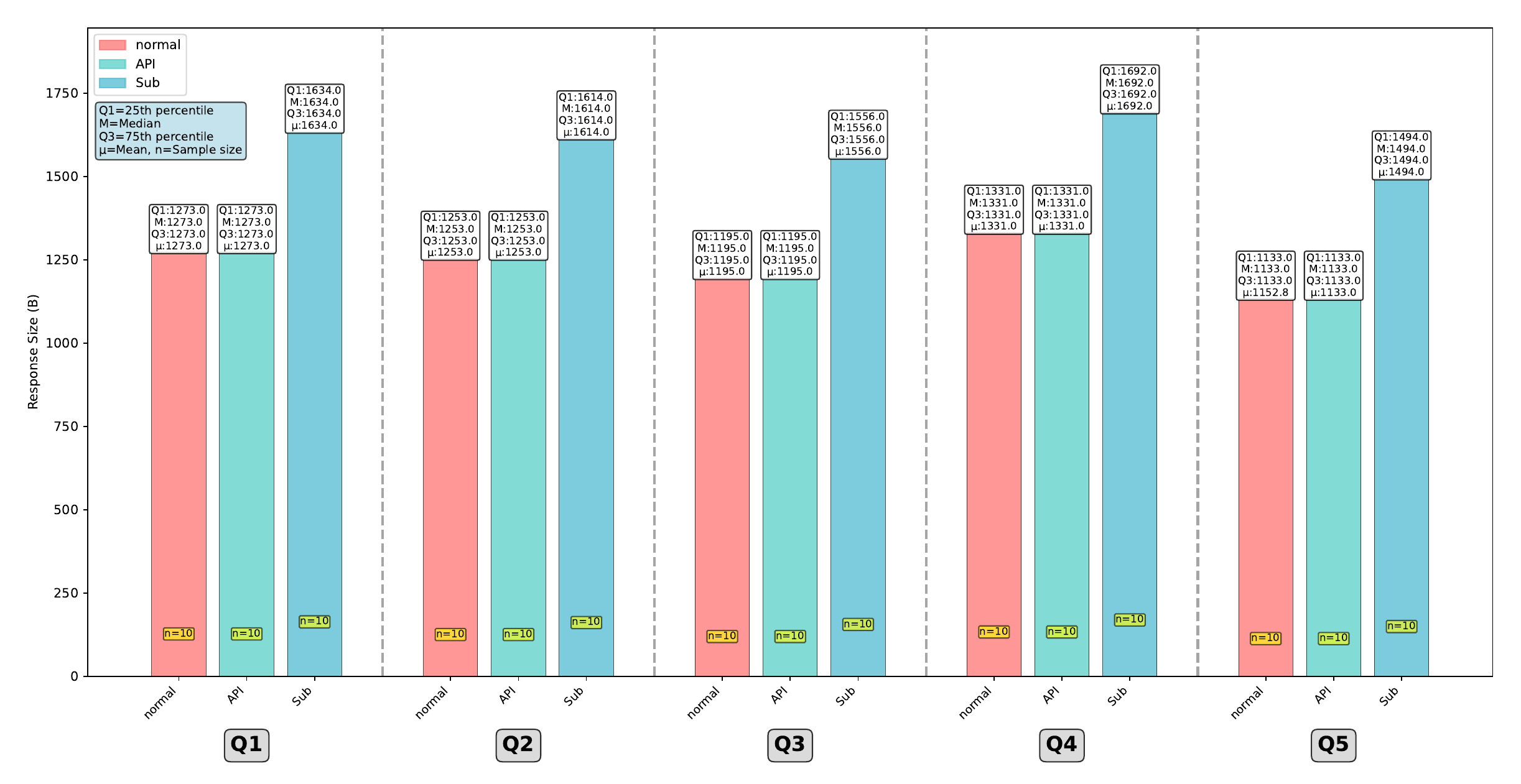}
	\caption{Q1 - Q5 communication response size comparison(DeepSeek-R1-Distill-Llama-8B).}
	\label{Figure:Q1 - Q5 Response Size Comparison(8B)}
\end{figure}

\begin{figure}[!htbp]
	\centering
	\includegraphics [width=0.8\linewidth ]{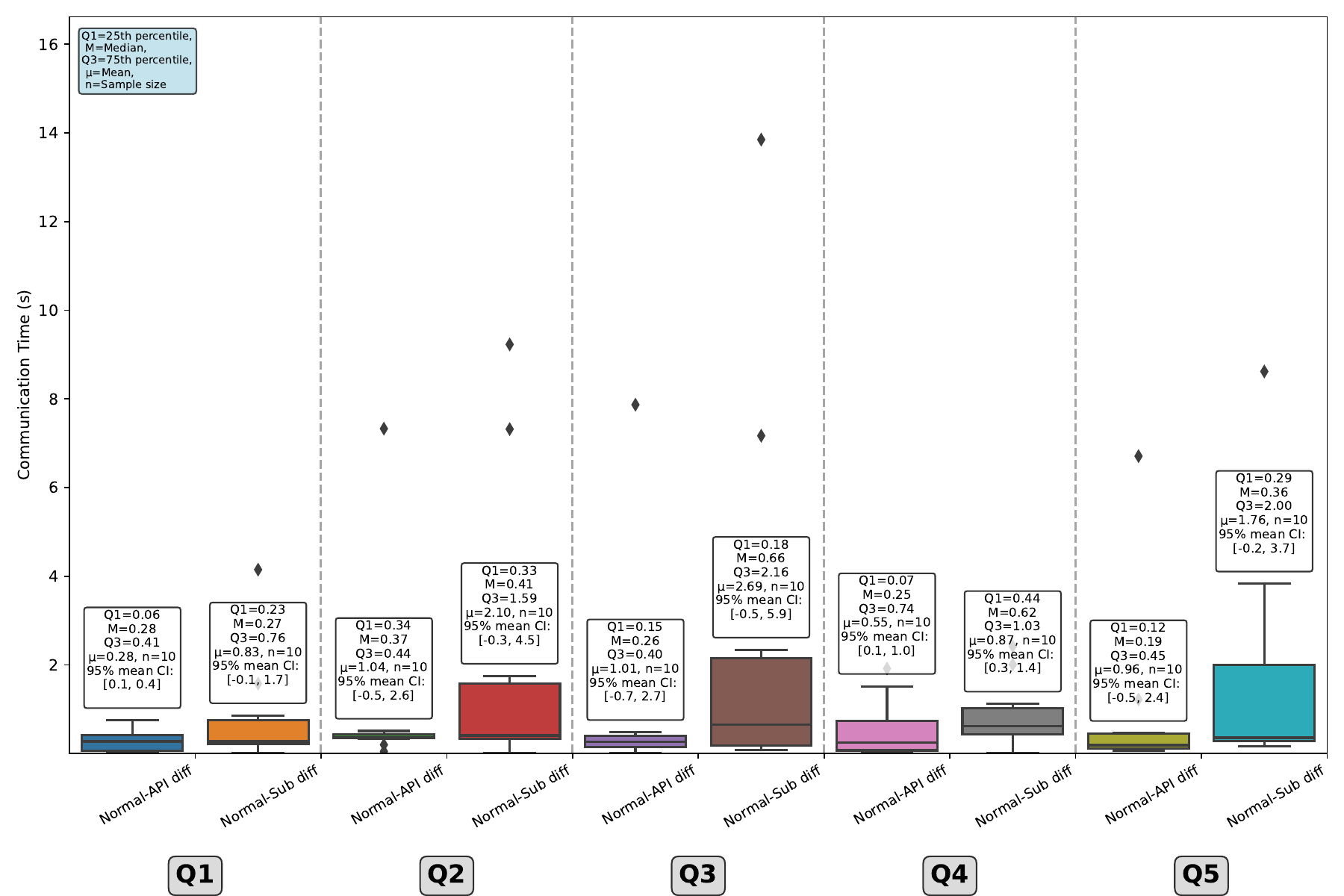}
	\caption{Q1 - Q5 communication time differentials comparison(DeepSeek-R1-Distill-Llama-8B).}
	\label{Figure:Q1 - Q5 Communication Time differentials Comparison(8B)}
\end{figure}

\begin{figure}[!htbp]
	\centering
	\includegraphics [width=0.8\linewidth ]{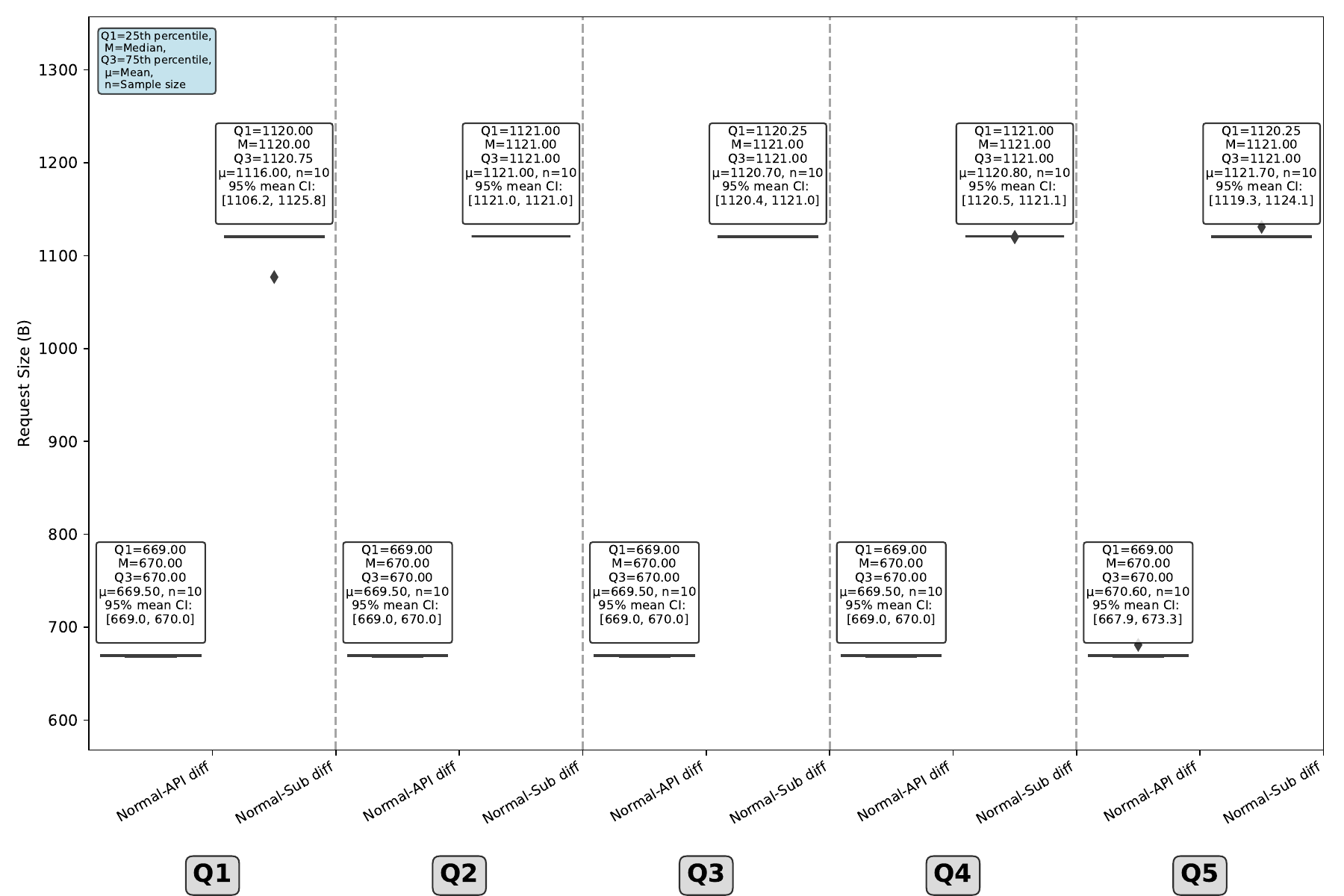}
	\caption{Q1 - Q5 communication request size differentials comparison(DeepSeek-R1-Distill-Llama-8B).}
	\label{Figure:Q1 - Q5 Request Size differentials Comparison(8B)}
\end{figure}

\begin{figure}[!htbp]
	\centering
	\includegraphics [width=0.8\linewidth ]{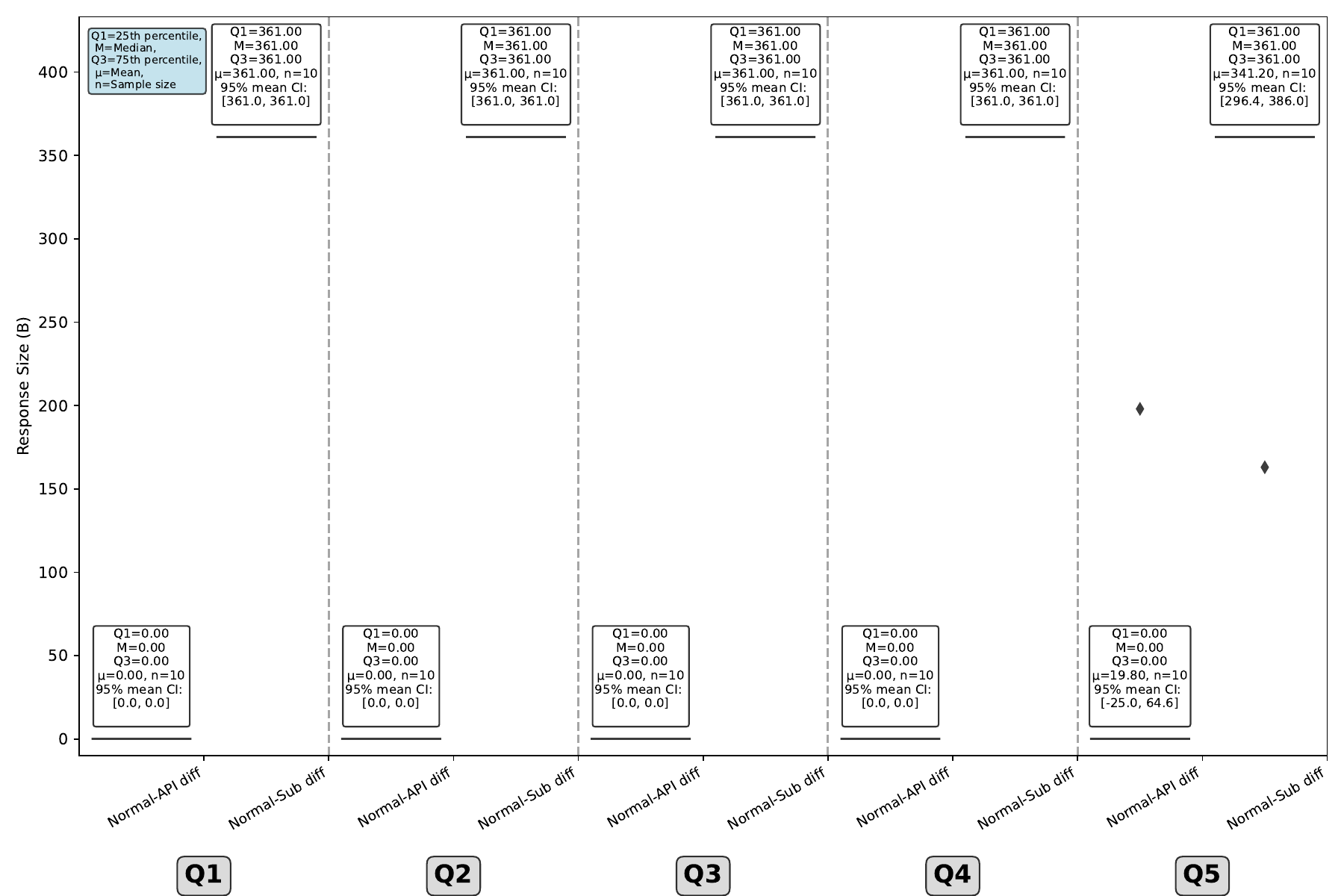}
	\caption{Q1 - Q5 communication response size differentials comparison(DeepSeek-R1-Distill-Llama-8B).}
	\label{Figure:Q1 - Q5 Response Size differentials Comparison(8B)}
\end{figure}

\newpage
 \section{Long Question}
 \label{app2}

LQ1:\{"text":"Please read the following text and answer the questions below.\texttt{\textbackslash n} <text> A leaked internal memo revealed that the company’s decision to reduce middle management layers was initially aimed at speeding up communication. However, the change unexpectedly resulted in the creation of new task forces and informal leaders, increasing both efficiency and confusion across departments. The CEO later described the process as “a self-correcting evolution rather than a planned reform.” </text>\texttt{\textbackslash n} What is the correct answer to this question: Which statement best summarizes the cause-and-effect chain described in the passage?\texttt{\textbackslash n} Choices:\texttt{\textbackslash n} (A) The restructuring created spontaneous leadership dynamics that altered communication patterns.\texttt{\textbackslash n} (B) The reorganization completely failed and caused total communication breakdown.\texttt{\textbackslash n} (C) The change was a preplanned top-down reform from the start.\texttt{\textbackslash n}(D) The CEO took no notice of any internal changes.\texttt{\textbackslash n} Let' s think step by step:"\}

LQ2:\{"text":"Please read the following text and answer the questions below. \texttt{\textbackslash n} <text> A regional agricultural survey from 2023 found that demand for fresh produce rose 15\% during harvest months but fell by almost the same margin in off-season periods. The authors emphasized that these figures showed correlation rather than causation, as consumer preferences were also influenced by festival schedules and marketing campaigns. </text>\texttt{\textbackslash n}What is the correct answer to this question: What conclusion does the author explicitly support about consumption and seasonality?\texttt{\textbackslash n}Choices:\texttt{\textbackslash n}(A) Consumption directly causes changes in supply.\texttt{\textbackslash n}(B) There is a correlation between consumption and seasonal supply, not a proven causal link.\texttt{\textbackslash n}(C) Seasonal factors have no effect on consumer behavior.\texttt{\textbackslash n}(D) Marketing alone determines all consumption patterns.\texttt{\textbackslash n}\texttt{\textbackslash n}Let’s think step by step:\texttt{\textbackslash n}"\}

LQ3:\{"text":"Please read the following text and answer the questions below.\texttt{\textbackslash n}<text> During the community board meeting, members debated a new parking policy. After two hours of discussion, a vote was called: 7 supported the motion, 6 opposed, and 2 abstained. The chairperson declared the motion approved but suggested revisiting it in three months to evaluate public feedback. </text>\texttt{\textbackslash n}What is the correct answer to this question: What was the factual outcome of the meeting?\texttt{\textbackslash n}Choices:\texttt{\textbackslash n}(A) The proposal was approved by a narrow vote margin.\texttt{\textbackslash n}(B) The proposal was unanimously rejected.\texttt{\textbackslash n}(C) No decision was reached and voting was postponed.\texttt{\textbackslash n}(D) The chairperson vetoed the result after the vote.\texttt{\textbackslash n}\texttt{\textbackslash n}Let’s think step by step:\texttt{\textbackslash n}"\}

LQ4:\{"text":"Please read the following text and answer the questions below.\texttt{\textbackslash n}<text> In this study, researchers combined online questionnaires, semi-structured interviews, and field observations to understand urban commuting habits. The participants were limited to residents of a single metropolitan district. The authors noted that the small sample size restricted generalization of results. </text>\texttt{\textbackslash n}What is the correct answer to this question: Which option correctly restates the methodological stance described?\texttt{\textbackslash n} Choices: \texttt{\textbackslash n} (A) A mixed-methods design was used within one urban cohort, acknowledging limited generalizability.\texttt{\textbackslash n}(B) The study relied only on laboratory experiments across multiple cities.\texttt{\textbackslash n}(C) The research was entirely theoretical with no empirical basis.\texttt{\textbackslash n}(D) The authors used meta-analysis of global datasets.\texttt{\textbackslash n}\texttt{\textbackslash n}Let’s think step by step:\texttt{\textbackslash n}"\}

LQ5:\{"text":"Please read the following text and answer the questions below.\texttt{\textbackslash n}<text> When discussing the proposed fuel tax, the minister acknowledged that short-term dissatisfaction was likely. However, he emphasized that preliminary polls suggested public sentiment would stabilize once citizens saw the environmental benefits. He warned against assuming an immediate political crisis. </text>\texttt{\textbackslash n}What is the correct answer to this question: What short-term public reaction did the speaker predict?\texttt{\textbackslash n}Choices:\texttt{\textbackslash n}(A) Initial discontent that would later calm as benefits became visible.\texttt{\textbackslash n}(B) Instant and irreversible political collapse.\texttt{\textbackslash n}(C) Total public indifference from the beginning.\texttt{\textbackslash n}(D) Violent protests lasting several years.\texttt{\textbackslash n}\texttt{\textbackslash n}Let’s think step by step:\texttt{\textbackslash n}"\}


%




\ifCLASSOPTIONcaptionsoff
  \newpage
\fi

\end{document}